\documentclass[journal,10pt,twocolumn]{IEEEtran}

\usepackage[utf8]{inputenc}
\usepackage{amsmath,bm,lipsum}
\usepackage{psfrag}
\usepackage{amsfonts,mathtools,amssymb,subcaption}
\usepackage[font=small]{caption}
\usepackage{amsmath,amssymb,amsthm,mathtools}%
\usepackage[ruled,linesnumbered]{algorithm2e}
\usepackage{flafter}
\usepackage[numbers,sort&compress]{natbib}
\makeatletter
\newcommand*{\rom}[1]{\expandafter\@slowromancap\romannumeral #1@}
\makeatother

\newtheorem{remark}{Remark}
\newcommand{\x}{\bm{x}}
\newcommand{\y}{\bm{y}}
\newcommand{\A}{\bm{A}}
\newcommand{\W}{\bm{W}}
\newcommand{\I}{\bm{I}}
\renewcommand{\H}{\bm{H}}
\renewcommand{\L}{\bm{L}}
\renewcommand{\S}{\bm{S}}
\newcommand{\U}{\bm{U}}
\newcommand{\ub}{\bm{u}}
\newcommand{\xb}{\bm{x}}
\newcommand{\bLambda}{\bm{\Lambda}}

\newcommand{\Ocal}{\mathcal{O}}
\newcommand*{\transp}{\mathsf{T}}
\newcommand*{\herm}{\mathsf{H}}
\newcommand{\hx}{\hat{\bm{x}}}
\newcommand{\hy}{\hat{\bm{y}}}

\newcommand{\bphi}{\bm{\phi}}
\newcommand{\xs}[1]{\mathcal{#1}}

\def\Phib{\bm\Phi}

\def\Phia{\Phib_{\A}}

\def\Tb{\bm T}
\def\Hb{\bm H}
\def\Bb{\bm B}

\def\bb{\bm b}
\def\Ub{\bm U}
\def\Sb{\bm S}
\def\Hb{\bm H}
\def\J{\mathcal{J}}
\def\Bs{\Bb_{\Ub}^{\A}}
\def\bs{\bb_{\Ub,i}^{\A}}

\def\Js{\J_{\Ub}^{\A}}
\def\A{\mathcal{A}}
\def\alphab{\boldsymbol \alpha}
\def\lambdab{\boldsymbol \lambda}

\newcommand{\Hev}{\bm H_{\rm ev}}

\newcommand{\card}[1]{\vert\mathcal{#1}\vert}
\newcommand{\diagg}[1]{{\rm diag}(#1)}

\newcommand{\sran}[1]{{\rm span}\{#1\}}
\newcommand{\snull}[1]{{\rm null}\{#1\}}

\newtheorem{mydef}{Definition}%
\newtheorem{myprop}{Proposition}%
\usepackage{color}

\renewcommand{\baselinestretch}{0.98} 
\begin{document}

\title{Advances in Distributed Graph Filtering}

\author{Mario~Coutino, \emph{Student Member, IEEE}, Elvin~Isufi,  \emph{Student Member, IEEE},  \\Geert~Leus, \emph{Fellow, IEEE}% <-this % stops a space
\thanks{The authors are with the faculty of Electrical Engineering, Mathematics and Computer Science, Delft University of Technology, 2826 CD Delft, The Netherlands. E-mails: {$\{$m.a.coutinominguez, e.isufi-1, g.j.t.leus$\}$@tudelft.nl}. This research is supported in part by the ASPIRE project (project 14926 within the STW OTP programme), financed by the Netherlands Organization for Scientific Research (NWO). Mario Coutino is partially supported by CONACYT. A conference version of this work was presented in \cite{coutino2017distributed}.}
}

%\markboth{Draft - Technical Report to submit to Information-Theoretic Methods in Data Acquisition, Analysis, and Processing}

\maketitle

\begin{abstract}
%The ultimate proof of understanding the signal evolution over a graph is reflected in our ability to filter it.

Graph filters are one of the core tools in graph signal processing. A central aspect of them is their direct distributed implementation. However, the filtering performance is often traded with distributed communication and computational savings. To improve this tradeoff, this work generalizes state-of-the-art distributed graph filters to filters where every node weights the signal of its neighbors with different values while keeping the aggregation operation linear. This new implementation, labeled as edge-variant graph filter, yields a significant reduction in terms of communication rounds while preserving the approximation accuracy. In addition, we characterize the subset of shift-invariant graph filters that can be described with edge-variant recursions. By using a low-dimensional parametrization the proposed graph filters provide insights in approximating linear operators through the succession and composition of local operators, i.e., fixed support matrices, which span applications beyond the field of graph signal processing. A set of numerical results shows the benefits of the edge-variant filters over current methods and illustrates their potential to a wider range of applications than graph filtering.
%
%
%The analytical and numerical results presented in this paper illustrate the potential and benefits of the general family of edge-variant graph filters.
\end{abstract}

%===============================================================================
\section{Introduction}
\label{sec:int}

\IEEEPARstart{F}{iltering} is one of the core operations in signal processing. The necessity to process large amounts of data defined over non-traditional domains characterized by a graph triggers advanced signal processing of the complex data relations embedded in that graph. Examples of the latter include biological, social, and transportation network data. The field of graph signal processing (GSP)~\cite{taubiny2000geometric, shuman2013emerging,sandryhaila2013discrete} has been established to incorporate the underlying structure in the processing techniques.

Through a formal definition of the graph Fourier transform (GFT), harmonic analysis tools employed for filtering in traditional signal processing have been adapted to deal with signals defined over graphs~\cite{taubin1996optimal,shuman2011distributed,narang2013signal,onuki2016graph,segarra2017optimal,isufi2017autoregressive}. Similarly to time-domain filtering, graph filters manipulate the signal by selectively amplifying/attenuating its graph Fourier coefficients. Graph filters have seen use in applications including signal analysis~\cite{shuman2016vertex, sandryhaila2014discrete}, classification~\cite{belkin2006manifold,ma2016diffusion}, reconstruction~\cite{narang2013signal,girault2014semi,IsufiDistWien18}, denoising~\cite{Zhang2008,onuki2016graph, yaugan2016spectral, isufi2017distributedSparse} and clustering~\cite{tremblay2016compressive}. Furthermore, they are the central block in graph filterbanks~\cite{tay2015design,teke2017extending}, wavelets~\cite{hammond2011wavelets}, and convolutional neural networks~\cite{defferrard2016convolutional,gama2018convolutional}.

Distributed graph filter implementations emerged as a way to deal with the ubiquity of big data applications and to improve the scalability of computation. By allowing nodes to exchange only local information,  finite impulse response (FIR)~\cite{shuman2011distributed, sandryhaila2013discrete, segarra2017optimal} and infinite impulse response (IIR)~\cite{shi2015infinite, isufi2017autoregressive} architectures have been devised to implement a variety of responses.

However, being inspired by time domain filters, the above implementations do not fully exploit the structure in the data. The successive signal aggregations are locally weighted with similar weights often leading to high orders in approximating the desired response. To overcome this challenge, this paper proposes a generalization of the distributed graph filtering concept by applying edge-based weights to the information coming from different neighbors. While the detailed contributions are provided in Section~\ref{subsec:ov_contr}, let us here highlight that the above twist yields in graph filters that are flexible enough to capture complex responses with much lower complexity.

%As a result, in this work, our main focus is on developing a general family of graph filters that: $(i)$ are flexible enough to capture complex responses with reduced complexity, and $(ii)$ maintain the distributable nature of classical graph filters.

\subsection{Related Works}

Driven by the practical need to implement a linear function with few local operations, the works in~\cite{segarra2017optimal,isufi2017FiltRnd} have put efforts to ease the communication and computational costs of graph filters (GF).

In~\cite{segarra2017optimal}, the authors modified the polynomial graph filters (i.e., the FIR structure) to graph filters with node-dependent weights. This architecture, referred to as a node-variant (NV) FIR graph filter, assigns different weights to different nodes and yields the same distributable implementation as the classical FIR graph filter~\cite{shuman2011distributed,sandryhaila2013discrete}. The NV FIR filter addresses a broader family of linear operators (e.g., analog network coding) that goes beyond the class of shift-invariant graph filters. However, the NV FIR filter uses the same weight for all signals arriving at a particular node, ignoring the affinity between neighbors. As we show next, this limits the ability of the NV FIR filter in approximating the desired linear operator with very low orders.

%In practice, for both centralized and distributed implementations, high-order graph filters are often required to properly approximate the desired filter response for particular tasks. Hence, this situation leads to high computational and implementation costs. In the literature, several works have aimed to reduce the implementation complexity by modifying classical graph filters. 
%They refer to this kind of graph filters as node-variant (NV) graph filters. However, despite the fact that NV graph filters apply different weights to distinct nodes, they use the same weight for all neighboring signals of a particular node~\cite{segarra2017optimal}. Thus, they fail to fully leverage the network structure.  

The work in~\cite{isufi2017FiltRnd} introduced stochastic sparsification to reduce the cost of a distributable implementation. Here, the authors considered random edge sampling in each aggregation step to implement the filter output with a lower complexity. Although conceptually similar to this work, the filter following~\cite{isufi2017FiltRnd} is stochastic and, therefore, the results hold only in expectation. Moreover, since this approach applies only to shift invariant filters, such as the FIR filter~\cite{shuman2011distributed, sandryhaila2013discrete} and the IIR~\cite{isufi2017autoregressive} implementations, it cannot address linear operators that are not shift invariant.

Another related problem, which can be interpreted like graph filtering, is the multilayer sparse approximation of matrices~\cite{Magoarou2016}. Different from the previous two approaches, here a dense linear transform (matrix) is approximated through a sequence of sparse matrix multiplications to obtain a computational speedup. While this framework can be considered as sequential diffusions over a network, the support of such sparse matrices differs in each iteration. This in practice can be a limitation since it often requires information from non-adjacent nodes within an iteration. Finally, in~\cite{barbarossa2009distributed} the problem of optimal subspace projection by local interactions was studied. This paper proposed an algorithm to design the weights of a network to achieve the fastest convergence rate for this kind of linear operators. Although their method provides optimal weights for projecting the data to a particular subspace, it does not address the GSP setup of interest: implementation of graph filters or general linear operators.

\subsection{Paper Contributions}\label{subsec:ov_contr}

The main contribution of this work is the extension of the state-of-the-art graph filters to edge-variant (EV) graph filters. Due to the increased degrees of freedom (DoF), these filters allow for a complexity reduction of the distributed implementation while maintaining the approximation accuracy of current approaches. The salient points that broaden the existing literature are listed below.

%Considering previous work in the field, in this paper, we focus on extending state-of-the-art graph filters to edge-variant graph filters in order to obtain a complexity reduction of the distributed filter implementation and a better approximation of arbitrary filter responses. The main contributions of this paper that broaden the existing literature are listed below.
\begin{itemize}
	\item[--] We present edge-variant architectures to implement FIR and IIR graph filtering. This framework extends the state-of-the-art graph filters by allowing nodes to weigh differently the information coming from different neighbors. In this way, only local exchanges are needed for each shift, thus yielding an efficient distributable implementation. Three forms are analyzed: First, the general class of linear edge-variant FIR filters is presented and its distributable implementation is discussed. Then, following the per-tone equalization idea~\cite{van2001per}, the constrained edge-variant FIR graph filter is introduced. This filter maintains a similar distributable implementation as the general form, yet allowing a simple least-squares design. Finally, the family of edge-variant autoregressive moving average graph filters of order one (ARMA$_1$) is treated. This new IIR distributable architecture allows a better trade-off between approximation accuracy and convergence rate than current approaches.
	\item[--] Through the definition of the filter modal response, we give a Fourier interpretation to a particular family of edge-variant graph filters. This subfamily shows a shift-invariant nature and links the filtering operation with the scaling applied on the graph modes (e.g., the graph shift eigenvectors).
	\item[--] Besides outperforming state-of-the-art graph filters in GSP tasks such as approximating a user-provided frequency response, distributed consensus, and Tikhonov denoising, we present two new applications that could be addressed distributively with the proposed edge-variant graph filter. The latter includes a distributed solution of an inverse problem and distributed beamforming.
\end{itemize}

\subsection{Outline and Notation}
This paper is organized as follows: Section~\ref{sec:prem} reviews the preliminaries of GSP, distributed graph filtering, and further defines the modal response of a graph filter. Section~\ref{sec:fir} generalizes the FIR graph filters to the edge-variant version. Here, we introduce the shift-invariant edge-variant graph filter and characterize its graph modal response. Section~\ref{sec.CFIRev} analyzes a particular subfamily of edge-variant FIR graph filters, which enjoys a similar distributed implementation and a least-squares design strategy. In Section~\ref{sec:iir}, we generalize the idea of edge-variant filtering to the class of IIR graph filters. Section~\ref{sec:num} corroborates our findings with numerical results and Section~\ref{sec:con} concludes this paper.

In this paper, we adopt the following notation. Scalars, vectors, matrices, and sets are denoted by lowercase letters $(x)$, lowercase boldface letters $(\bm x)$, uppercase boldface letters $(\bm X)$, and calligraphic letters $(\mathcal{X})$, respectively. $[\bm X]_{ij}$ denotes the $(i,j)$th entry of the matrix $\bm X$ whereas $[\bm x]_i$ represents the $i$th entry of the vector $\bm x$. $\bm X^\transp$, $\bm X^\herm$, and $\bm X^{-1}$ are respectively the transpose, the Hermitian, and inverse of $\bm X$. The Moore-Penrose pseudoinverse of $\bm X$ is $\bm X^{\dagger}$. The Khatri-Rao product between $\bm X$ and $\bm Y$ is written as $\bm X \ast \bm Y$, while their Hadamard product as $\bm X \odot \bm Y$. $\bm 1$ and $\bm I$ are the all-one vector and identity matrix of appropriate size, respectively. ${\rm vec}(\cdot)$ is the vectorization operation. ${\rm diag}(\cdot)$ refers to a diagoal matrix with its argument on the main diagonal. ${\rm null}\{\cdot\}$ and ${\rm span}\{\cdot\}$ denote the nullspace and span of their argument. ${\rm nnz}(\bm X)$ and ${\rm supp}\{\bm X\}$ are the number of nonzero entries and the support of $\bm X$. Finally, we define the set $[K] = \{1,2,\ldots,K\}$.
%===============================================================================
\section{Preliminaries}
\label{sec:prem}

This section recalls the preliminary material that will be useful in the rest of the paper. It starts with the definition of the graph Fourier transform (GFT) and graph filtering. Then, two distributed recursions that implement FIR and IIR filtering operations on graphs are presented. Finally, the modal response of a graph filter is defined.

\textbf{Graph Fourier transform.} Consider an $N$-dimensional signal ${\bm x}$ residing on the vertices of a graph $\mathcal{G} = (\mathcal{V},\mathcal{E})$ with $\mathcal{V} = \{v_1, \ldots, v_N\}$ the set of $N$ vertices and $\mathcal{E} \subseteq \mathcal{V}\times\mathcal{V}$ the set of $M$ edges. Let $\W$ be the weighted graph adjacency matrix with $\W_{i,j} \neq 0$ if $(v_j,v_i) \in \mathcal{E}$ and $\W_{i,j} = 0$, otherwise. For an undirected graph, the graph Laplacian matrix is $\L$. Both $\W$ and $\L$ are valid candidates for the so-called graph shift operator $\S$, an $N\times N$ matrix that carries the notion of delay in the graph setting \cite{taubiny2000geometric, taubin1996optimal,sandryhaila2013discrete,shuman2013emerging}. Given the decomposition $\S = \U\bLambda\U^{-1}$ (assuming it exists), the GFT of $\x$ is defined as the projection of $\x$ onto the modes of $\S$, i.e., $\hx = \U^{-1}\x$. Likewise, the inverse GFT is $\x = \U\hx$. Following the GSP convention, the eigenvalues $\bLambda = {\rm diag}(\lambda_1, \ldots, \lambda_N)$ are referred to as the {graph frequencies}.

\textbf{Graph filtering.} A linear shift-invariant graph filter is an operation on the graph signal with graph frequency domain output
\begin{equation}
\label{eq.gfilt}
\hy = h(\bLambda)\hx.
\end{equation}
Here, $h({\boldsymbol \Lambda})$ is a diagonal matrix with the filter frequency response on its diagonal. More formally, the frequency response of a graph filter is a function 
\begin{equation}\label{eq.freq_resp}
h : \mathbb{C} \mapsto \mathbb{R},\;\;\; \lambda_i \rightarrow h(\lambda_i)
\end{equation}
that assigns a particular value $h(\lambda_i)$ to each graph frequency $\lambda_i$. This definition is akin to the one used in traditional signal processing, however depending on the underlying topology some shift operators might not be simple, i.e., the multiplicity of some eigenvalues is greater than one. So, there is no one-to-one mapping between the graph frequencies $\lambda_i$ and the graph modes $\bm{u}_i$. For this reason, at the end of this section, we will introduce the notion of graph modal response which treats the graph filters from the graph shift eigenvector perspective. Finally, by applying the inverse GFT on both sides of \eqref{eq.gfilt}, we have the vertex domain filter output
\begin{equation}
\label{eq.gfilt_vx}
\y = \H\x,
\end{equation}
with $\H = \U h(\bLambda)\U^{-1}$. 

\textbf{FIR graph filters.} A popular form of $\H$ is its expression as a polynomial of the graph shift operator \cite{taubin1996optimal,shuman2011distributed,sandryhaila2013discrete}, i.e.,
\begin{equation}
\label{eq.FIR}
\H_{\text{c}} \triangleq \sum_{k = 0}^K\phi_k\S^k,
\end{equation}
which we refer to as the \emph{classical} FIR graph filter. It is possible to run the FIR filter \eqref{eq.FIR} distributively due to the locality of $\S$~\cite{shuman2011distributed,segarra2017optimal}. In particular, since $\S^k\x = \S(\S^{k-1}\x)$ the nodes can compute locally the $k$th shift of $\x$ from the former $(k-1)$th shift. Overall, an FIR filter of order $K$ requires $K$ local exchanges between neighbors and amounts to a computational and communication complexity of $\Ocal(MK)$.

To expand the possible set of operations that can be implemented distributively through FIR recursions, \cite{segarra2017optimal} proposed the NV FIR graph filter. These filters have the node domain form
\begin{equation}
\label{eq.NV_FIR}
\H_{\text{nv}} \triangleq \sum_{k = 0}^K{\rm diag}(\bphi_k)\S^k,
\end{equation}
where the vector $\bphi_k = [\phi_{k,1}, \ldots, \phi_{k,N}]^T$ contains the node dependent coefficients applied at the $k$th shift. For $\bphi_k = \phi_k\mathbf{1}$, the NV FIR filter \eqref{eq.NV_FIR} reduces to the classical FIR filter \eqref{eq.FIR}. The NV FIR filter preserves also the efficient implementation of \eqref{eq.FIR} since it applies the node coefficients to the $k$th shifted input $\S^k\x = \S(\S^{k-1}\x)$ with a computational complexity of $\Ocal(MK)$. 

If a linear operator $\tilde{\H}$ needs to be approximated by a matrix polynomial as in~\eqref{eq.FIR}, the filter order $K$ can become large if a high accuracy is required. As the computational complexity scales with $K$, large-order graph filters incur high costs. The NV graph filters provide a first approach to tackle this issue. Starting from Section~\ref{sec:fir}, we generalize these ideas towards an \emph{edge-variant} (EV) graph filter alternative, which due to its enhanced DoF can approximate $\tilde{\H}$ with even a lower order $K$. Therefore, it leads to a more efficient implementation. One of the main benefits of both the NV and the EV graph filters is that they address a broader class of operators $\tilde{\H}$ which not necessarily share the eigenvectors with $\S$, such as the analog network coding \cite{segarra2017optimal}.

\textbf{IIR graph filters.} In \cite{isufi2017autoregressive}, the authors introduced an ARMA recursion on graphs to implement distributively IIR graph filtering, i.e., a filtering operation characterized by a rational frequency response. The building block of this filter is the so-called ARMA graph filter of order one (ARMA$_1$). This filter is obtained as the steady-state of the first-order recursion
\begin{equation}\label{eq:ARMA}
  \y_{t} = \psi\S\y_{t-1} + \varphi\x,
\end{equation}
with arbitrary $\y_0$ and scalar coefficients $\psi$ and $\varphi$. The operation~\eqref{eq:ARMA} is a distributed recursion on graphs, where neighbors now exchange their former output $\y_{t-1}$ rather than the input $\x$. The per-iteration complexity of such a recursion is $\Ocal(M)$. Given $\psi$ satisfies the convergence conditions for \eqref{eq:ARMA} \cite{isufi2017autoregressive}, the steady-state output of teh ARMA$_1$ is
\begin{eqnarray}\label{eq:ARMAss}
\y &\triangleq& \underset{t\rightarrow\infty}{\lim}\y_t = \varphi\sum\limits_{\tau=0}^{\infty}(\psi\S)^{\tau}\bm x = \varphi(\I - \psi\S)^{-1}\x \\
&\triangleq& \H_{\rm arma_{1}}\x\nonumber.
\end{eqnarray}
Such a filter 
%has the steady-state frequency response 
%
%\begin{equation}\label{eq:ARMAClas}
%  h(\lambda) = \frac{\varphi}{1 - \psi\lambda},
%\end{equation}
%
%and 
addresses several GSP tasks including Tikhonov denoising, graph signal interpolation under smoothness prior \cite{isufi2017autoregressive}, and aggregate graph signal diffusion \cite{isufi2017distributed}. In Section \ref{sec:iir}, we extend \eqref{eq:ARMA} to an edge-variant implementation with the aim to improve its convergence speed without heavily affecting the approximation accuracy.
 
\textbf{Graph modal response.} Before moving to the main contributions of this work, we define next the \emph{modal response} of a graph filter. The latter represents the scaling that the graph modes experience when a graph signal undergoes a linear shift-invariant graph filtering operation.

\begin{mydef}{\textnormal{(Graph modal response)}}
\label{def:modal}
The modal response of a linear shift-invariant graph filter 
\begin{equation}
	\bm H = \bm U {\rm diag}(h_1,\ldots,h_N) \bm U^{-1},
\end{equation}
is defined as the function
$$h : [N] \rightarrow \mathbb{C},\;\;\; i \mapsto h_i,$$
where $h_i$ is the scaling experienced by the $i$th graph mode.
\end{mydef}

This definition is equivalent to the graph frequency response \eqref{eq.freq_resp} when the shift operator has a simple spectrum. Since this is not always the case, we feel that the graph modal response is closer in meaning to the classical frequency response, and use it in the rest of the paper.
%===============================================================================
\section{Edge-Variant FIR Graph Filters}
\label{sec:fir}

Let us assume a scenario in which each node \emph{trusts} differently the information coming from different neighbors, e.g., a person is likely to weigh more the opinion of his/her partner than that of a colleague on a personal recommendation. So, it is reasonable to treat this case as a graph filter, where each node weighs differently the information of its neighbors.

Here, we formalize the above intuition in terms of EV FIR graph filters. We first introduce the general form of these filters while in Section~\ref{sec.sievFIR} we focus on the class of shift-invariant edge-variant (SIEV) FIR graph filter. The filter design strategy is discussed in Section~\ref{sec.FIR-EV_design}.

\subsection{General Form}
Consider an extension of the above edge-dependent fusion to several diffusion steps (signal shifts) where in each shift a different set of weights is used. At the $k$th diffusion, node $v_i$ weighs its neighbouring node $v_l$ with the weight $\phi_{i,l}^{(k)}$. Hence, in each shift $k\in[K]$, and for each node $v_i$, there is a set of coefficients $\{\phi_{i,l}^{(k)}\}$ for $l \in \xs{N}_{v_i}$. Here, $\xs{N}_{v_i}$ denotes the set of nodes adjacent to $v_i$, and $K$ is the number of shifts. Mathematically, the above behavior can be written through an order-$K$ \emph{general} EV FIR graph filter defined as
\begin{eqnarray}
\begin{split}
\bm{H}_{\rm{ev}} &\triangleq \bm\Phi_1 + \bm\Phi_2\bm\Phi_1 + \ldots + \bm\Phi_K\bm\Phi_{K-1}\cdots\bm\Phi_1 \\
 &=  \sum\limits_{k=1}^{K}\bm\Phi_{k:1},
\end{split}
\label{eq:evfilt}
\end{eqnarray}
where $\bm\Phi_{k:1} = \bm\Phi_k\bm\Phi_{k-1}\cdots\bm\Phi_1$ and $\bm\Phi_{j}\in\mathbb{C}^{N\times N}$ is an edge-weighting matrix constructed from the coefficient set $\{\{\phi_{1,l}^{(j)}\},\cdots,\{\phi_{N,l}^{(j)}\}\}$, more specifically $[\bm\Phi_j]_{il}=\phi_{i,l}^{(j)}$. By construction, the support of $\bm\Phi_j$ and $\bm S+\bm I$ is the same $\forall\;j\in[K]$. Since $\bm S$ might have zero entries on its diagonal, i.e., $\bm S = \bm W$, here we extend the support of $\{\bm\Phi_{j}\}_{j\in[K]}$ to allow each node to use also its own information. Note that definition \eqref{eq:evfilt} does not impose any symmetry on the coefficient matrices $\bm \Phi_{j}$. Depending on how adjacent nodes trust each other, the applied weights can be different.

The filter can differently be interpreted through time-varying shift operators~\cite{kalofolias2017learning,kolar2010estimating}, where $\bm\Phi_j$ is the weighted, possibly directed shift operator for the $j$th diffusion step with the support of $\bm S+\bm I$. Therefore, the general EV FIR filter accounts for signals that are generated through time-varying systems in directed subgraphs of the original graph. Here, the filter coefficient matrix only allows for \emph{edge deletion} or a re-weighting of graph flows.

Note that recursion~\eqref{eq:evfilt} is a distributed graph filter. To compute the output $\bm y = \Hev \bm x$, each node is only required to track the following quantities:
\begin{itemize}
	\item the shifted signal output $\bm x^{(k)} = \bm\Phi_k\bm x^{(k-1)}, \bm x^{(0)} = \bm x$,
	\item the accumulator output $\bm y^{(k)} = \bm y^{(k-1)} + \bm x^{(k)}, \bm y^{(0)} = \bm 0$.
\end{itemize}
Both these operations can be computed locally in each node by combining only neighboring data. Hence,~\eqref{eq:evfilt} preserves the efficient distributed implementation of the classical FIR graph filter~\eqref{eq.FIR} with a complexity of $\Ocal(MK)$.

Before addressing the design strategy of the filter~\eqref{eq:evfilt}, in the sequel, we introduce a particular structure of EV FIR graph filters that enjoy a graph Fourier domain interpretation.

\subsection{Shift-Invariant Edge-Variant Graph Eigenfilters}\label{sec.sievFIR}

An important family of graph filters is that of \emph{shift-invariant graph filters}, i.e., filters that commute with the graph shift operator $\Sb$. That is, given the shift $\Sb$ and the filter matrix $\Hb$, the following holds
\begin{equation}\label{eq:shiftInv}
	\Sb\Hb = \Hb\Sb.
\end{equation}
For a non-defective shift operator $\Sb$ and filter $\Hb$, i.e., the matrices accept an eigenvalue decomposition, condition \eqref{eq:shiftInv} is equivalent to saying that the matrices $\Sb$ and $\Hb$ are jointly diagonalizable, or that their eigenbases coincide.

There is no reason to believe that the graph filters of form~\eqref{eq:evfilt} are shift invariant. However, it is possible to characterize a subset of edge-variant graph filters that satisfy this property. To do that, we rely on the following assumptions:  
\vskip.25cm
{\renewcommand\labelitemi{}
\begin{itemize}
	\item{$(A.0)$} $\bm S$ is diagonalizable;
	\item{$(A.1)$} Each $\bm\Phi_j,\, j\in[K]$ is diagonalizable with the eigenbasis of $\bm S$;
	\item{$(A.2)$} Each $\bm\Phi_j,\, j\in[K]$ shares the support with $\bm S + \bm I$.
\end{itemize}
}

Given the above assumptions hold, we can rewrite~\eqref{eq:evfilt} as
%By considering matrices $\{\bm\Phi_j\}_{j\in\{0\}\cup[K]}$ meeting $(A.1)$, it is clear that using the construction~\eqref{eq:evfilt} leads to a shift-invariant graph filter, i.e.,
\begin{eqnarray}
\begin{split}
	\bm H_{\rm ev} =\sum\limits_{k=1}^{K}\bm\Phi_{k:1} = \bm U\bigg[\sum\limits_{k=1}^{K}\prod\limits_{j=1}^{k}\bm\Lambda_j\bigg]\bm U^{-1},
	\end{split}
\end{eqnarray}
where we substituted $\bm\Phi_j = \bm U \bm\Lambda_j\bm U^{-1}$. To provide a closed-form expression for the effect of such graph filters on the graph modes, let us first describe the set of \emph{fixed-support matrices that are diagonalizable with a particular eigenbasis} (i.e., matrices that meet $(A.1)$ and $(A.2)$).
%
%Therefore, in order to characterize the subset of graph filters of the form~\eqref{eq:evfilt} meeting the previous assumptions, we first describe the set of \emph{fixed-support matrices that are diagonalizable using a particular eigenbasis} which are the set of matrices meeting $(A.1)$ and $(A.2)$. Then, we provide a closed form expression for the effect of such graph filters in the graph modes. 
%
Mathematically, this set is defined as
%Mathematically, the set of matrices meeting $(A.1)$ and $(A.2)$ can be formally defined as the set
%
\begin{equation}\label{eq:fMtx}
  \J_{\Ub}^{\A} = \{ \bm A : \bm A = \Ub\bm\Omega\Ub^{-1}, [{\rm vec}(\bm A)]_i = 0,\;\forall\;i\in\A\},
\end{equation}
where $\A$ is the index set defining the zero entries of $\bm S + \bm I$ and $\bm\Omega$ is diagonal. The fixed-support condition in $\Js$ can be expressed in the linear system form
\begin{equation}\label{eq:linS}
  \Phib_{\A}{\rm vec}(\bm A) = {\boldsymbol 0},
\end{equation}
with $\Phib_{\A}\in\{0,1\}^{\card{A}\times N^2}$ denoting the selection matrix whose rows are the rows of an $N^2\times N^2$ identity matrix indexed by the set $\A$. By leveraging the vectorization operation properties and the knowledge of the eigenbasis of $\bm A$, we can rewrite~\eqref{eq:linS} as
\begin{equation}\label{eq:linK}
  \Phib_{\A}{\rm vec}(\bm A) = \Phia(\Ub^{-\transp}\ast\Ub)\bm\omega = {\boldsymbol 0},
\end{equation}
where ``$\ast$" represents the Kathri-Rao product and $\bm\omega = [[\bm\Omega]_{11},[\bm\Omega]_{22},\ldots,[\bm\Omega]_{NN}]^\transp$ is the vector containing the eigenvalues of $\bm A$. From~\eqref{eq:linK}, we see that $\bm\omega$ characterizes the intersection of the nullspace of $\Phia$ and the range of $\Ub^{-\transp}\ast\Ub$. More formally, we write
%
%Therefore, we can formally express this property of the eigenvalue vector as
\begin{equation}\label{eq:Tb}
  \bm\omega \in \snull{\Tb_{\Ub}^{\A}},
\end{equation}
with $\Tb_{\Ub}^{\A} = \Phib_\A(\Ub^{-\transp}\ast\Ub)$.

With this in place, the following proposition characterizes the matrices that belong to the set $\Js$.
%
%Using the above property of the eigenvalue vector, $\lambdab$, the matrices that belong to the set $\Js$ can be characterized by the following theorem.
\begin{myprop}{\textnormal{(Graph shift nullspace property)}}\label{prop:nullSpace}
  Given an orthonormal basis $\Ub$ and a sparsity pattern defined by the set $\A$, the matrices within the set $\J_{\Ub}^{\A}$ are of the form $\bm A = \Ub\bm\Omega\Ub^{-1}$ and have eigenvalues given by
  \begin{equation}\label{eq:eigS}
    \bm\Omega = {\rm diag}(\boldsymbol B_{\Ub}^{\A}\alphab),
  \end{equation}
  where the matrix $\Bb_{\Ub}^{\A}$ is a basis for the nullspace of $\Tb_{\Ub}^{\A}$, i.e.,
  $$
  \sran{\Bb_{\Ub}^{\A}} = {\rm null}\{\Tb_{\Ub}^{\A}\},
  $$
  and $\alphab$ is the basis expansion coefficient vector.
  \label{eq:prop1}
\end{myprop}
\begin{proof}
  The proof follows from~\eqref{eq:linK}-\eqref{eq:Tb}.
\end{proof}

The above result is not entirely surprising and has been used for assessing the uniqueness of the graph shift operator in topology identification~\cite{coutino2018sparsest}. Here, we leverage Proposition~\ref{eq:prop1} for interpreting the response of the SIEV graph filters. More specifically, under $(A.1)$ and $(A.2)$ we can express each matrix $\bm \Phi_j$ of~\eqref{eq:evfilt} as
\begin{equation}
\bm\Phi_j = \Ub\diagg{\Bs\alphab_j}\Ub^{-1},
\end{equation}
and write any SIEV FIR filter as
%\footnote{Note that a SIEV FIR graph filter has a posynomial-like structure, i..e, it is a function of the form $f(x_1,\ldots,x_n) = \sum\limits_{k=1}^{K}c_kx_1^{a_{1k}}\cdots x_n^{a_{nk}}$.}
\begin{equation}
  \Hb_{\rm siev} = \Ub\bigg[ \sum\limits_{k=1}^{K}\prod_{j=1}^{k}\diagg{\Bs\alphab_j}\bigg]\Ub^{-1}.
  \label{eq:evFreq}
\end{equation}
The following proposition formally characterizes the frequency interpretation of such filters in terms of the modal response.

\begin{myprop}{\textnormal{(Modal Response of SIEV FIR)}} \label{prop_sievMR}
An FIR graph filter of the form~\eqref{eq:evfilt} satisfying $(A.1)$ and $(A.2)$ has $i$th modal response
\begin{equation}\label{eq:evEV_h}
  h_i = \sum\limits_{k=1}^{K}\prod\limits_{j=1}^{k}(\bb_{\Ub,i}^{\A})^{\transp}\alphab_{j} + (\bb_{\Ub,i}^{\A})^\transp\alphab_0,
\end{equation}
where $(\bb_{\Ub,i}^{\A})^\transp$ is the $i$th row of $\Bs$.
\end{myprop}
\begin{proof}
  The proof follows directly from~\eqref{eq:evFreq}.
\end{proof}

An interesting outcome from Proposition~\ref{prop_sievMR} is that {the filter response is independent of the graph frequencies}. This is clear from \eqref{eq:evEV_h}, where we see that the eigenvalue $\lambda_i$ does not appear in the expression of $h_i$. Therefore, we can interpret the SIEV FIR graph filters as \emph{eigenvector filters}, since they act on the eigenmodes of the graph. That is, for each graph eigenmode (eigenvector) $\ub_i$, $\Hb_{\rm siev}$ might apply a different gain given by~\eqref{eq:evEV_h} (independent of $\lambda_i$) to the component of the input signal $\xb$ in the direction of $\ub_i$.
%
%So, we can formally describe the response, with respect to the graph eigenmodes, of the restricted family of shift-invariant edge-variant graph filters (SIEV-GFs) in the following proposition. 
%We refer to this subset of graph filters as \emph{shift-invariant edge-variant graph eigenfilters} (SIEV-GEs) due to the fact that they directly act on the eigenmodes of the graphs, as their response does not depend on the eigenvalues of the shift operator $\bm S$. 
%
%
%From this proposition, an interesting property of the filters of the form~\eqref{eq:evFreq} can be observed: \emph{the filter response is independent of the graph frequencies}. This is clear from~\eqref{eq:evEV_h}, where it is seen that the eigenvalue $\lambda_i$ does not appear in the modal response $h_i$. Therefore, we could interpret this family of filters as \emph{eigenvector filters}. That is, for the $i$th eigenvector $\ub_i$, the filter $\Hb_{\rm siev}$ applies a fixed gain given by~\eqref{eq:evEV_h} (independent of $\lambda_i$) to the component of the input signal $\xb$ in the direction of $\ub_i$. 
%
%As a result, the filter, $\Hb_{\rm siev}$ can be treated as a filter that applies a different gain to each of the eigenmodes of the graph. 
This is in contrast to classical FIR graph filters which apply the same polynomial expression to all modes $\{\bm u_i\}_{i\in[N]}$.

The following section introduces methods for designing EV FIR graph filters in the node domain and SIEV FIR graph filters using the parametrization in~\eqref{eq:evEV_h}.
%\textcolor{red}{Elvin: "reduced-order parametrization" might not be clear to the reader what you mean. Consider rephrasing this.}

\subsection{Filter Design}\label{sec.FIR-EV_design}

{\textbf{General form.}}
{%\color{blue}
Given a desired operator $\tilde{\bm H}$, we design an EV FIR filter $\bm H_{\rm{ev}}$ [cf.~\eqref{eq:evfilt}] that approximates $\tilde{\bm H}$ as the solution of the optimization problem
\begin{equation}\label{eq:genmin}
\begin{aligned}
& \underset{\{\bm{\Phi}_{k}\}}{\text{minimize}}
& & \Vert\tilde{\bm H} - \sum\limits_{k=1}^{K}\bm\Phi_{k:1}\Vert \\
& \text{subject to}
& & \bm\Phi_{k:1} = \bm\Phi_k\bm\Phi_{k-1}\cdots\bm\Phi_1,\\
& & & {\rm supp}\{\bm \Phi_k\} = {\rm supp}\{\bm S + \bm I\} \;\forall\; k\in[K],
\end{aligned}
\end{equation}
where $\Vert\cdot\Vert$ is an appropriate distance measure, e.g., the Frobenius norm ($\Vert\cdot\Vert_{F}$), or the spectral norm ($\Vert\cdot\Vert_2$).% In this work, we mainly focus on the Frobenius norm to measure the approximation quality.

Unfortunately,~\eqref{eq:genmin} is a high-dimensional nonconvex problem and hard to optimize. An approach to finding a local solution for it is through block coordinate methods, which provide local convergence guarantees when applied to such problems~\cite{xu2013block}. In fact, the cost in~\eqref{eq:genmin} is a \emph{block multi-convex} function, i.e., the cost function is a convex function of $\Phib_i$ with all the other variables fixed.

Starting then with an initial set of matrices $\{\Phib_j^{(0)}\}_{j\in[K]}$ (potentially initialized with an order-$K$ classical FIR filter), we solve a sequence of optimization problems where at the $i$th step, the matrix $\Phib_i$ is found. That is, at the $i$th iteration, we fix the matrices $\{\Phib_j^{(0)}\}_{j\in[K]\backslash \{i\}}$ and solve the convex problem
\begin{equation}\label{eq:genMinStep}
\begin{aligned}
& \underset{\Phib_i}{\text{minimize}}
& &\Vert\tilde{\bm H} - \sum\limits_{k=1}^{K}\Phib_{k:(i+1)}^{(0)}\Phib_i\Phib_{(i-1):1}^{(0)}\Vert \\
& \text{subject to}
& &{\rm supp}\{\bm \Phi_i\} = {\rm supp}\{\bm S + \bm I\},
\end{aligned}
\end{equation}
where $\Phib_{a:b}^{(0)} = \Phib_{a}^{(0)}\Phib_{a-1}^{(0)}\ldots\Phib_{b+1}^{(0)}\Phib_{b}^{(0)}$ for $a \ge b$ and $\Phib_{a:b}^{(0)} = \I$, otherwise. Then, the matrix $\Phib_i^{(0)}$ is updated with its solution and the procedure is repeated for all $\{\Phib_j\}_{j\in[K]}$. If the final fitting error is large, the whole process can be repeated until the desired performance is reached, or until a local minimum is found.

%\textcolor{blue}{\textcolor{red}{The blue part is from Mario, for discussion.}Therefore, starting with an initial set of matrices $\{\Phib_j^{(0)}\}_{j\in\{0\}\cup[K]}$ (possibly initialized with the shift matrix $\bm S$), we solve a sequence of optimization where at the $i$th step the pair of matrices $\{\Phib_0,\Phib_i\}$ are found. That is, at the $i$th iteration, we fix the matrices $\{\Phib_j^{(0)}\}_{j\not\in\{0,i\}}$ and solve the following convex problem:
%\begin{equation}\label{eq:genMinStep}
%	\begin{array}{lc}
%		\underset{\Phib_0,\Phib_i}{\text{minimize}} & \Vert\tilde{\bm H} - \sum\limits_{k=1}^{K}(\prod\limits_{j > i}^{k}\Phib_j^{(i)})\Phib_i(\prod\limits_{j\neq 0,j<i}^{i-1}\bm \Phi_j^{(i)}) - \bm \Phi_0\Vert \\
%		\text{subject to} & {\rm supp}\{\bm \Phi_i\} = {\rm supp}\{\bm S + \bm I\}
%	\end{array}.
%\end{equation}
%Then, the matrices $\Phib_i^{(0)}$ and $\Phi_0^{(0)}$ are updated with the solution of~\eqref{eq:genMinStep}. This procedure is repeated for each of the matrices $\{\Phib_j\}_{j\in[K]}$ until all matrices have been updated. If extra rounds are required, due to the current fitting error, the whole process can be repeated until the desired performance is reached, or the method convergences to a local minima.}

Although filter \eqref{eq:evfilt} is the most general EV FIR filter form, the non-convexity encountered in the above design strategy may often lead to a local solution with an unacceptable performance. To tackle such issue, in Section~\ref{sec.CFIRev}, we introduce a constrained EV FIR filter which provides a higher flexibility than the state-of-the-art graph filters while accepting a simple least squares design.

}

{\textbf{SIEV form.}}
Besides enjoying the modal response interpretation, the SIEV FIR filter also has a simpler design than the general form~\eqref{eq:evfilt}. For $\{\tilde{h}_i\}_{i=1}^{N}$ being the desired graph modal response\footnote{This can be for instance a low-pass form if we want to keep only the eigenvector contribution associated with the low graph frequencies.}, the 
%
%Let $\{\tilde{h}_i\}_{i=1}^{N}$ denote the graph modal response that we want to implement distributively, e.g., a low-pass form if we want to keep only the eigenvector contribution associated with the low graph frequencies. The 
%
SIEV FIR filter design consists of
% fitting the modal response~\eqref{eq:evEV_h} to $\{\tilde{h}_i\}_{i=1}^{N}$ by 
solving the optimization problem
%
%When we restrict ourselves to the case of SIEV-GEs, the design has a simpler form than the one for general EV graph filters [cf.~\eqref{eq:genmin}]. However, the design problem does not have a suitable convex formulation that allows for an efficient solution. 
%
%More specifically, consider a desired graph modal response $\{\tilde{h}_i\}_{i=1}^{N}$, as well as the expression of the modal response of a SIEV-GE [cf.~\eqref{eq:evEV_h}]. In order to design the SIEV graph filter, we require to solve the following optimization problem:
\begin{equation}\label{eq:cstMinSIEV}
  \begin{array}{ll}
    \underset{\{\alphab_j\}}{\rm minimize} & \sum\limits_{i=1}^{N}\big\Vert \tilde{h}_{i} - \sum\limits_{k=1}^{K}\prod\limits_{j=1}^{k}(\bb_{\Ub,i}^{\A})^{T}\alphab_{j} \big\Vert_2^{2}.
  \end{array}
\end{equation}
Similarly to \eqref{eq:genmin}, problem~\eqref{eq:cstMinSIEV} is nonconvex and cannot in general be solved up to global optimality with standard convex optimization methods. However, \eqref{eq:cstMinSIEV} is also a block multi-convex function in each $\alphab_i$, $i\in[K]$ individually and, therefore, the block coordinate descent methods~\cite{xu2013block} can be employed to find a local minimum. Alternatively, the straightforward analytical expression of the gradient of the cost function allows the use of off-the-shelf solvers for global optimization, such as the MATLAB's built-in \texttt{fmincon} function~\cite{matlabFmincon}.

%to the general EV-GF design, the cost function in~\eqref{eq:cstMinSIEV} is not convex and can not be generally solved up to global optimality with standard convex optimization methods. However, noticing that~\eqref{eq:cstMinSIEV} is also a block multi-convex function, i.e., for each $\alphab_i$, $i=0,1,\ldots,K$, the cost function is a convex function of $\alphab_i$ while all the other variables, $\alphab_j$, $j\neq i$ are kept fixed, we can employ block coordinate descent methods for its optimization with guaranteed local convergence. Furthermore, if off-the-shelf solvers for global optimization are preferred due to the straightforward analytical expression for the gradient of the cost function, common functions such as the MATLAB's built-in \texttt{fmincon} function~\cite{matlabFmincon} can be employed.

% ----
%Considering these facts, in the next section we propose a constrained version of the EV graph filter that provides a simpler design for the weighting matrices and still preserves the efficient implementation with complexity $\Ocal(MK)$.

%A closely related kind of graph filters are the ones that naturally exhibit a \emph{posynomial structure}. In the following, a brief digression is made to show the relation of this graph filter with the SIEV graph filters.

%\subsection{Posynomial Graph Filters}

\section{Constrained Edge-Variant FIR Graph Filters}\label{sec.CFIRev}

To overcome the design issues of the general EV FIR filter, here we present a constrained version of it that retains both the distributed implementation and the edge-dependent weighting. This reduction of the DoF will, in fact, allow us to design the filter coefficients in a least squares fashion. The structure of these filters along with their distributed implementation is presented in the next section. In Section~\ref{subsec:sicev} we provide a modal response interpretation of these filters, while in Section~\ref{subsec:CEV_desi} we present the design strategy.

\subsection{General Form}\label{subsec:gfCEV}

The \emph{constrained} EV (CEV) FIR graph filter is defined as
\begin{equation}
\H_{\rm{cev}} = \Phib_1 + \Phib_2\S + \cdots + \Phib_K\S^{K-1} \triangleq \sum\limits_{k=1}^{K}\bm \Phi_k  \bm S^{k-1},
\label{eq:evD}
\end{equation}
where the edge-weighting matrices $\{\bm \Phi_k\}_{k\in[K]}$ again share the support with $\bm S + \bm I$.
These filters enjoy the same distributed implementation of the general form \eqref{eq:evfilt}. In fact, each node can compute locally the filter output by tracking the following quantities:
\begin{itemize}
	\item the regular shift output $\bm x^{(k)}=\bm S \bm x^{(k-1)},~\bm x^{(0)} = \bm x$,
	\item the weighted shift output  $\bm z^{(k)} = \bm \Phi_k\bm x^{(k-1)}$,
	\item the accumulator output $\bm y^{(k)} = \bm y^{(k-1)} + \bm z^{(k)},~\bm y^{(0)} = \bm 0$.
\end{itemize}
From the locality of $\bm S$ and $\Phib_k$, both $\bm x^{(k)}$ and $\bm z^{(k)}$ require only neighboring information. The final filter output is $\bm y = \bm y^{(K)}$ which yields the same computational complexity of $\mathcal{O}(MK)$.

Note that construction~\eqref{eq:evD} still applies different weights to the signal coming from different edges. However, instead of adopting a different \emph{diffusion} matrix at every step, the signal diffusion occurs through the graph shift $\bm S$. The additional extra step mixes locally $\bm x^{(k-1)}$ using edge-dependent weights, which are allowed to vary for each $k$. We here adopt the term constrained for this implementation from the observation that the diffusion is performed using only a single shift operator matrix. Fig.~1(a) visually illustrates the differences between the different graph filters analyzed so far.

\begin{remark}
The NV graph filter from \cite{segarra2017optimal} [cf.~\eqref{eq.NV_FIR}] is a particular case of the CEV graph filter. The local matrices $\{\bm\Phi_k\}_{k=1}^{K}$ are in fact substituted by diagonal matrices with distinct elements across their diagonals.
% every row of the matrix $\bm\Phi_k,\, k = 1,\ldots, K$ has equal non-zero elements, i.e.,
% \begin{equation}
% \begin{array}{ll}
% \H_{\rm{nv}} & =\sum\limits_{k=1}^K \big((\bm \phi_{k}\bm 1^T)\odot \S \big) \S^{k-1} + {\rm diag}(\bm \phi_{0}). \\
% \end{array}
% \end{equation}
\end{remark}

\subsection{Shift-Invariant Constrained Edge-Variant Graph Eigenfilters}\label{subsec:sicev}

Following the same lines of Section~\ref{sec.sievFIR}, we can use the set $\J_{\Ub}^{\A}$~\eqref{eq:fMtx} to characterize the graph modal response of the CEV FIR graph filter when the matrices $\{\bm\Phi_k\}_{k=1}^{K}$ satisfy $(A.1)$ and $(A.2)$.
%
%express the structure of the CEV graph filter when the condition $(A.1)$ and $(A.2)$ are satisfied by the matrices $\{\bm\Phi_k\}_{k=0}^{K}$ [cf.~\eqref{eq:evD}].
%\begin{equation}
%  (\Phib_k \odot \bm S) \bm S = \bm S (\Phib_k \odot \bm S),
%\end{equation}
%is met for each $k = 0,\ldots,K$. 
%
This subset of CEV FIR graph filters, which we refer to as {shift-invariant CEV (SICEV) FIR graph filters}, can again be expressed in terms of $\bm B_{\bm U}^{\mathcal{A}}$ and $\{\alphab_k\}_{k=0}^{K}$ as
\begin{equation}\label{eq:SICEV}
  \Hb_{\rm sicev} = \Ub\bigg[ \sum\limits_{k=1}^{K}\diagg{\Bs\alphab_k \odot \lambdab^{\odot (k-1)}} \bigg]\Ub^{-1},
\end{equation}
where $\lambdab^{\odot k}$ denotes the $k$th element-wise power of the eigenvalue vector of the shift operator $\bm S$. The subsequent proposition formalizes the modal response of these filters.

%
%To make the effect of this kind of filters on the graph frequency content more precise, in the following proposition, we provide a formal definition of the modal response of filters with the structure defined in~\eqref{eq:SICEV}.

\begin{myprop}{\textnormal{(Modal Response of SICEV FIR)}} An FIR graph filter of the form~\eqref{eq:evD} satisfying $(A.1)$ and $(A.2)$ has $i$th modal response
\begin{equation}\label{eq:evCEV_h}
  h_i = \sum\limits_{k=1}^{K}\gamma_{ik}\lambda_i^{(k-1)},
\end{equation}
where $\gamma_{ik} = (\bb_{\Ub,i}^{\A})^{T}\alphab_{k}$ is the $k$th polynomial coefficient for the $i$th graph frequency and $(\bb_{\Ub,i}^{\A})^T$ is the $i$th row of $\Bs$.
\end{myprop}
\begin{proof}
  The proof follows directly from~\eqref{eq:SICEV}.
\end{proof}

From \eqref{eq:evCEV_h}, we see that there is a substantial difference between the SICEV FIR filter and the more general SIEV FIR graph filters. Here, the modal response is a polynomial in the graph frequencies. This is similar as for the classical FIR filter \eqref{eq.FIR}, but now each frequency has a different set of coefficients. In other words, the modal response of the SICEV FIR filter is \emph{a mode-dependent polynomial}. For readers more familiar with traditional discrete-time processing, this behavior can be interpreted as applying different polynomial filters to each frequency bin (see e.g.,~\cite{van2001per}).

%Different from the modal response of the more general SIEV graph filter [cf.~\eqref{eq:evEV_h}], the frequency response is a polynomial in the graph frequencies. This is as for the classical graph filter but now each frequency has a different set of coefficients. 
%
%In other words, the modal response of the SICEV-GE is \emph{a mode-dependent polynomial}. That is, each mode is related to a polynomial which is applied to its associated graph frequency. In traditional discrete time processing this behavior can somehow be interpreted, as applying different polynomial filters to each of the frequency bins obtained after discrete Fourier transform (DFT) processing~\cite{van2001per}.

\begin{remark}
The particular form of the SICEV FIR filter allows it to match all shift-invariant polynomial responses of order $K$ and a subset of higher-order polynomials of order up to $N-1$. The latter property derives from the observation that {any shift-invariant graph filter} is a polynomial of the graph shift operator~\cite{sandryhaila2013discrete} and from the filter response in~\eqref{eq:evCEV_h}. In fact, the SICEV FIR filter is still a polynomial of the shift $\bm S$, though with a different polynomial response per graph frequency. This additional freedom extends the set of functions that can be approximated by a SICEV FIR filter of order $K$. Fig.~1(b) further illustrates the relation among different graph filters.
\end{remark}
% \begin{figure}[t]
%   \centering
%   %\includegraphics[width=0.5\textwidth]{figures/oneHopCoeffCEV.pdf}
%   \includegraphics[width=0.5\textwidth]{figures/cevgfSets.eps}
%   \caption{Relation between polynomial graph filters and CV-GFs.}
%   \label{fig:cevgfSets}
% \end{figure}
\begin{figure*}[t]
  \centering
\begin{subfigure}[c]{.6\textwidth}
  \centering
  \includegraphics[width=\textwidth]{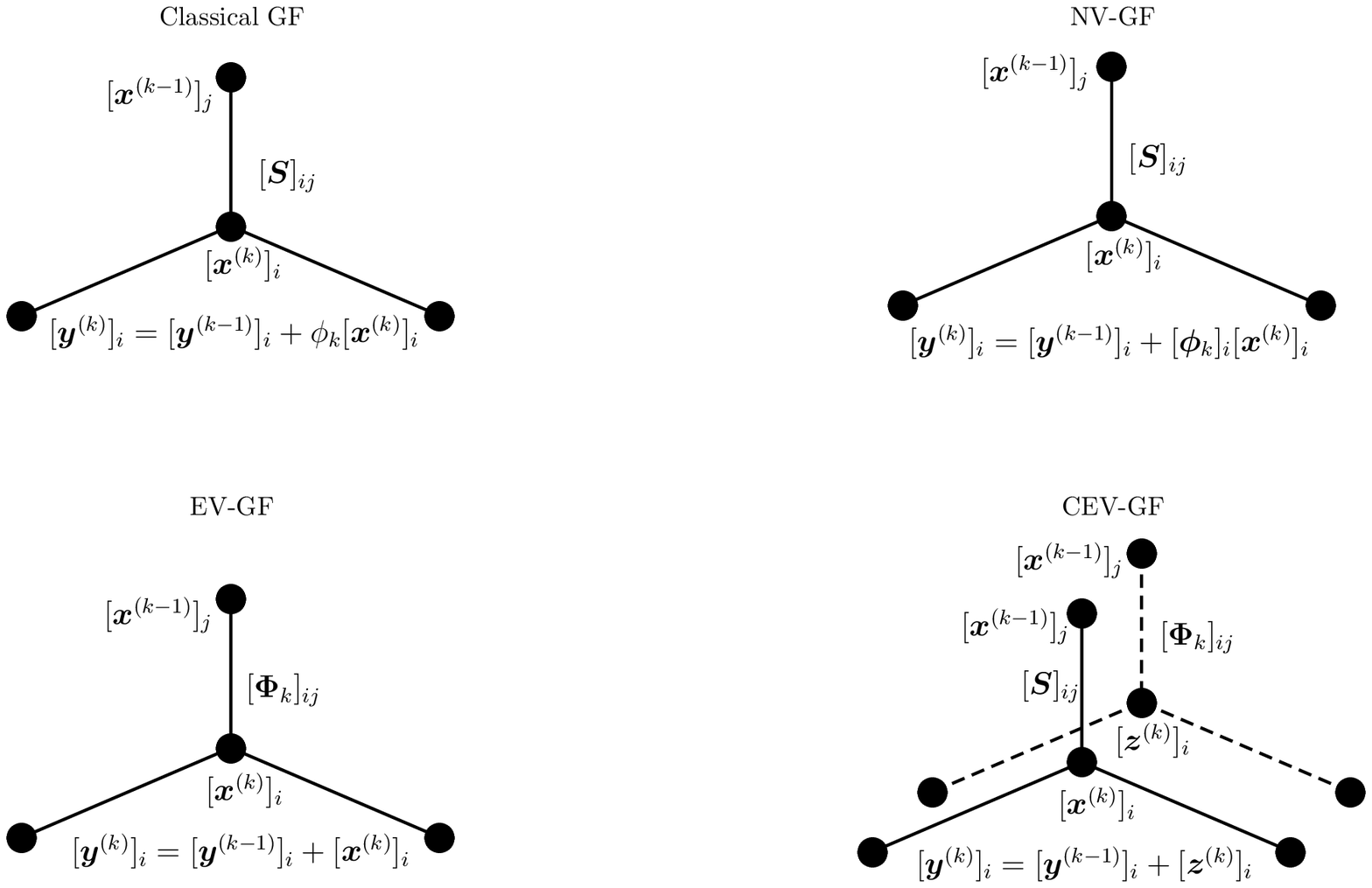}
  \caption{}
  \label{fig:cef1}
\end{subfigure}%
\centering
\begin{subfigure}[c]{.4\textwidth}
  \centering
  % \psfrag{Probability of Error}[Bc][c]{\scriptsize{Probability of Error $[P_{\rm{e}}]$}}
  \includegraphics[width=\textwidth]{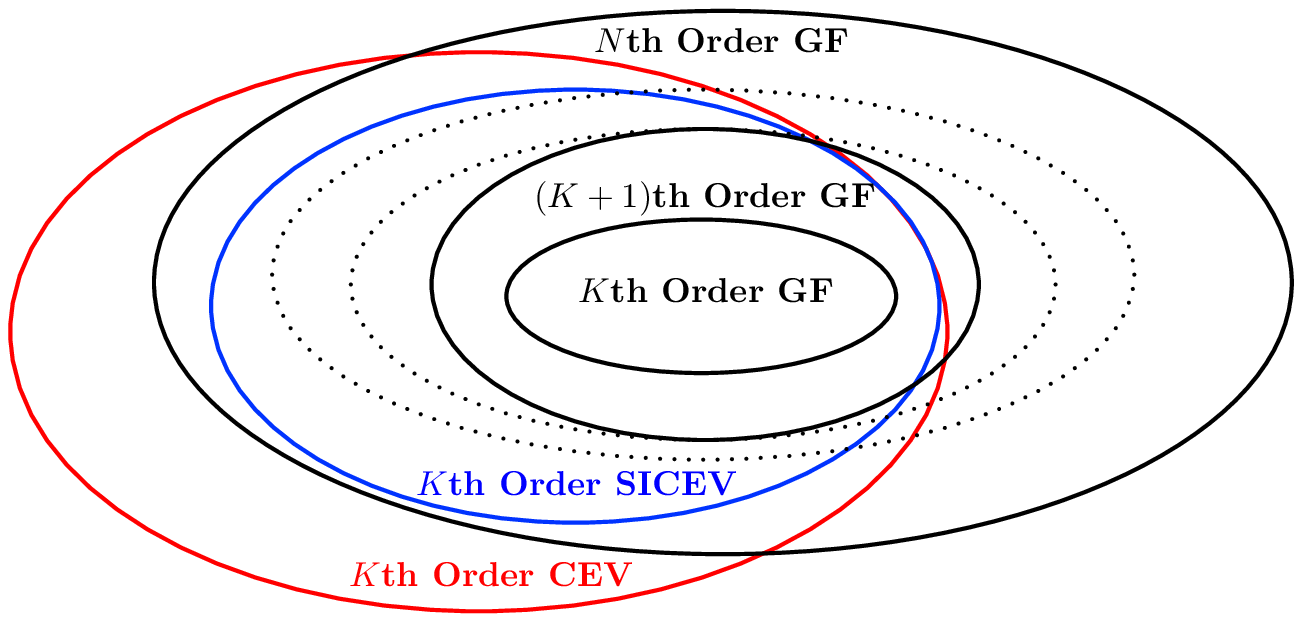}
  \caption{}
  \label{fig:cef2}
\end{subfigure}
	\caption{$(a)$ Illustration of the required transmission, scaling, and recursion performed by the different graph filters. (b) Relation between the classical and CEV FIR graph filters. This figure depicts the possibility of obtaining higher-order polynomial graph filters with reduced order CEV graph filters.}
	\label{fig:testcef}
\end{figure*}

\subsection{Filter Design}\label{subsec:CEV_desi}

%\subsubsection*{\textbf{General Form}}
%From the definition of the C-EV graph filter in~\eqref{eq:evD} it can be observed that it allows a distributed implementation. This can be seen due to the fact that every node tracks two results obtained by local interactions: $(i)$ the previous intermediate result $\bm x^{(k)}=\bm S\bm x^{(k-1)}$, and $(ii)$ the output of the filter $\bm y^{(k)} = \bm (\bm\Phi_k\odot\bm S)\bm x^{(k-1)}$, which are weighted combinations of neighbours data.

%From the construction in~\eqref{eq:evD} it is easy to see that the constrained version of the EV graph filter is distributed in nature since at each step the previous intermediate result, $\bm x^{(k)}=\bm S\bm x^{(k-1)}$, is locally weighted and combined at each node, i.e., $\bm y^{(k)} = \bm (\bm\Phi_k\odot\bm S)\bm x^{(k-1)}$, similarly to the classical and NV graph filters.

\textbf{General form.} Following a similar approach as in Section~\ref{sec.FIR-EV_design}, we can approximate a desired operator $\tilde{\bm H}$ with a CEV FIR filter by solving the problem
%
%Following the same approach as in the previous section, we can design the weighting matrices $\{\bm\Phi_k\}_{k=0}^{K}$ in the \emph{node domain} through a least squares approach:
\begin{equation}\label{eq:genminCEV}
\begin{aligned}
& \underset{\{\bm\Phi_k\}}{\text{minimize}}
& & \Vert\tilde{\bm H} - \sum\limits_{k=1}^{K}\bm \Phi_k  \bm S^{k-1}\Vert_F^2 \\
& \text{subject to}
& &{\rm supp}\{\bm \Phi_k\} = {\rm supp}\{\bm S + \bm I\} \;\forall\; k\in[K].
\end{aligned}
\end{equation}
Exploiting then the properties of the vectorization operator and the Frobenius norm, we can transform \eqref{eq:genminCEV} into
\begin{equation}\label{eq:genminCEV_vec}
\begin{aligned}
& \underset{\{\bm\phi_k\}}{\text{minimize}}
& & \Vert {\tilde{\bm h}} - \sum\limits_{k=1}^{K}(\S^{k-1} \otimes \bm I)\bm \phi_k\Vert_{2} \\
& \text{subject to}
& & \tilde{\bm h} \triangleq \text{vec}(\tilde{\bm\H}),~\bm \phi_k \triangleq \text{vec}(\bm \Phi_k),\\
& & &{\rm supp}\{\bm \Phi_k\} = {\rm supp}\{\bm S + \bm I\} \;\forall\; k\in[K].
\end{aligned}
\end{equation}
Since the support of the weighting matrices is known, problem \eqref{eq:genminCEV_vec} can be written in the reduced-size form
\begin{equation}\label{eq:genminCEV_vecc}
\begin{aligned}
& \underset{\{\bm\phi_k\}}{\text{minimize}}
& & \Vert \tilde{\bm h} -  \bm \Psi \bm \theta\Vert_{2}^2 \\
& \text{subject to}
& & \bm \Psi = [\bm I\; \check{\S}\;\cdots\;\check{\S}_K]\\
& & &\bm \theta = [\check{\bm \phi}_0^\transp\; \check{\bm \phi}_1^\transp\;\cdots \tilde{\bm \phi}_K^\transp]^\transp,
\end{aligned}
\end{equation}
%\in \mathbb{R}^{N^2\times {\rm nnz}(\bm S)\cdot K + N},\\
%\in \mathbb{R}^{{\rm nnz}(\S)\cdot K + N},
%
where, $\check{\bm\phi}_k$ is the vector $\bm\phi_k$ with the zero entries removed, $\check{\S}_k$ is the matrix $(\S^{k}\otimes \bm I)$ with the appropriate columns removed. In addition, if a regularized solution is desired, a natural penalization term might be the convex $\ell_1$-norm which induces sparsity in the solution yielding only few active coefficients.
%As the support of the weighting matrices is known, i.e., the locations of the nonzero entries are given by the nonzero entries of the shift operator, the filter (in its vectorized form) can be expressed as
%\begin{equation}
%\bm h_{\rm{c-ev}} = \sum\limits_{k=1}^{K}\tilde{\S}_{k}\tilde{\bm\phi}_k +\tilde{ \phi}_0,
%\end{equation}
%where ${\bm h}_{\rm{c-ev}} = \text{vec}({\H}_{\rm{c-ev}})$, $\tilde{\bm\phi}_k$ is the vector with the nonzero entries of $\bm\phi_k$,  and $\tilde{\S}_k$ is the matrix $(\S^{k-1}\otimes \bm I)\text{diag}(\bm s)$ with the zero columns removed. That is, as the support for each $\bm\Phi_k$ is known, only its nonzero entries must be estimated.
%%
%As a result, the optimal weights for the edges are given by the solution to the linear system
%\begin{eqnarray}
%\tilde{\bm h} &= \begin{bmatrix}
%\bm I & \tilde{\S}_1 & \ldots & \tilde{\S}_K
%\end{bmatrix}\begin{bmatrix}
%\tilde{\bm \phi}_0 \\
%\tilde{\bm \phi}_1 \\
%\vdots\\
%\tilde{\bm \phi}_K \\
%\end{bmatrix} = \bm \Psi \bm \theta,
%\label{eq:lnProb}
%\end{eqnarray}
%where $\bm \Psi \in \mathbb{R}^{N^2\times \rm{nnz}(\S)\cdot K + N}$ and $\bm \theta \in \mathbb{R}^{\rm{nnz}(\S)\cdot K + N}$. Here, $\rm{nnz}(\cdot)$ is the number of nonzero entries of a matrix. In addition, if a regularized solution is desired, a natural penalisation term might be the convex $\ell_1$-norm which induces sparsity in the solution leading to a reduced number of active coefficients.

Problem~\eqref{eq:genminCEV_vec} has a unique solution as long as $\bm\Psi$ is full column rank, i.e., ${\rm rank}(\bm\Psi) = {\rm nnz}(\S)\cdot K + N$. Otherwise, regularization must be used to obtain a unique solution.% for~\eqref{eq:lnProb}.
\begin{remark}
Besides leading to a simple least squares problem, the design of the CEV FIR filter can also be computed distributively. Given that each node knows the desired filter response and the graph shift operator (i.e., the network structure), it can be shown that by reordering the columns of $\bm\Psi$ and the entries of $\bm\theta$ the framework of splitting-over-features~\cite{manss2017distributed} can be employed for a decentralized estimation of $\bm\theta$.
\end{remark}
%\subsubsection{Distributed Coefficients Learning}

%\subsubsection*{\textbf{SICEV Graph Filter Design}}
\textbf{SICEV form.} Similar to the more general CEV FIR filter, the design of $\{\alphab_k\}_{k=1}^{K}$ for the SICEV form can be performed in a least squares fashion.

First, for a set of vectors $\{\alphab_k\}_{k=1}^{K}$ the modal response for the SICEV FIR filter reads as
\begin{equation}\label{eq.scievFIR_freps}
  \bm h_{ \lambda} = \sum\limits_{k=1}^{K} [\Bs\alphab_k \odot \lambdab^{\odot (k-1)}],
\end{equation}
where $\bm h_{\lambda}$ is obtained by stacking the modal responses, i.e., $\{h_i\}_{i=1}^{N}$, in a column vector. By using the properties of the Hadamard product, we can rewrite \eqref{eq.scievFIR_freps} as
\begin{eqnarray}
   \bm h_{\lambda} = \sum\limits_{k=1}^{K}\diagg{\lambdab^{\odot (k-1)}}\Bs\alphab_k = \sum\limits_{k=1}^{K}\bm M_k\alphab_k,
\end{eqnarray}
with $\bm M_k = \diagg{\lambdab^{\odot (k-1)}}\Bs$. Defining then $\bm M = [\bm M_1,\ldots,\bm M_K]$, and $\alphab = [\alphab_1^\transp,\ldots,\alphab_K^\transp]^\transp$, we obtain the linear relation
\begin{equation}\label{eq:Msyst}
   \bm h_{\lambda} = \bm M\alphab.
\end{equation}
Therefore, the approximation of a desired response $\tilde{\bm h}_{\bm \lambda}=~[\tilde{h}_1,\ldots,\tilde{h}_N]^\transp$ consists of solving the least squares problem
%
%to approximate a desired response $\tilde{\bm h}_{\bm \lambda}=~[\tilde{h}_1,\ldots,\tilde{h}_N]^T$ we can solve the following least squares problem
\begin{equation}
	\begin{array}{ll}
		\underset{\alphab\in\mathbb{R}^{d(K+1)}}{\text{minimize}} & \Vert \tilde{\bm h}_{\lambda} - \bm M \alphab \Vert_2
	\end{array}
\end{equation}
which has a unique solution when $\bm M$ is full column rank, i.e., ${\rm rank}(\bm M) = d(K+1) \leq N$.

\section{Edge-Variant IIR Graph Filters}
\label{sec:iir}

We now extend the edge-variant filtering concept to the class of IIR graph filters. As stated in Section~\ref{sec:prem}, we focus on the basic building block of these filters, i.e., the ARMA$_1$ recursion \eqref{eq:ARMA}. We follow the same organization of the former two sections, by introducing the edge-variant ARMA$_1$ structure in Section~\ref{eq.evARMA1stuc}, the shift-invariant version in Section~\ref{subsec:iirSIEV}, and the design strategies in Section~\ref{sub:armaDesg}.

\subsection{Edge-Variant ARMA$_1$}\label{eq.evARMA1stuc}
We build an edge-variant ARMA$_1$ recursion on graphs by modifying~\eqref{eq:ARMA} as
\begin{equation}\label{eq:ARMA_EV}
\y_t = \bm\Phi_1\y_{t-1} + \bm\Phi_0\x,
\end{equation}
where $\bm\Phi_0$ and $\Phib_1$ are the edge-weighting matrices having the support of $\bm S + \bm I$ that respectively weight locally the entries of $\y_{t-1}$ and $\x$. Proceeding similarly as in \cite{isufi2017autoregressive}, for $\|\Phib_1\|_2 < 1$, the steady-state output of \eqref{eq:ARMA_EV} is
\begin{equation}
  \y = \lim_{t \to \infty}\y_t = (\I - \Phib_1)^{-1}\bm\Phi_0\x \triangleq \H_{\rm eva_{1}}\x,\label{eq:ARMA_EV_ex}
\end{equation}
where now we notice the inverse relation w.r.t. the edge-weighting matrix $\Phib_1$. Recursion \eqref{eq:ARMA_EV} converges to \eqref{eq:ARMA_EV_ex} linearly with a rate governed by $\|\Phib_1\|_2$. The classical form \eqref{eq:ARMA} can be obtained by substituting $\Phib_1 = \psi\bm S$ and $\bm\Phi_0 = \varphi\I$.

The edge-variant ARMA$_1$ filter presents the same frequency interpretation challenges as the FIR filter counterpart. Therefore, we next analyze the shift-invariant version of it and we will see a rational modal response.

%Despite the fact that obtaining the filter frequency response of the edge-varying ARMA$_1$ graph filter $\H_{\rm evarma_{1}}$ proves to be a challenging task, similarly to the FIR case, we show how the response of such filters can be characterized for a subset of EV-ARMA$_{1}$ filters whose matrices $\{\Phib_1, \Phib_0\}$ are diagonalizable by the eigenbasis of the graph shift operator.

\subsection{Shift-Invariant EV ARMA$_1$}\label{subsec:iirSIEV}
By limiting the choices of $\{\Phib_0,\Phib_1\}$ to the one that satisfy $(A.1)$ and $(A.2)$, we obtain the shift-invariant edge-variant ARMA$_1$ (SIEVA$_1$) graph filter
%
%First, using Theorem~\ref{eq:th1} and considering that $(A.1)$ and $(A.2)$ are met for $\{\Phib_1,\Phib_0\}$, we obtain the following expression for~\eqref{eq:ARMA_EV_ex}:
\begin{equation}\label{eq:ARMA_exU}
  \H_{\rm sieva_{1}} = \Ub[(\I - \diagg{\Bs\alphab_1})^{-1}\diagg{\Bs\alphab_{0}}]\Ub^{-1},
\end{equation}
where $\alphab_{0}$ and $\alphab_1$ are the respective basis expansion vectors of $\Phib_0$ and $\Phib_1$ onto the nullspace of $\Tb_{\Ub}^{\A}$ (see Proposition~\ref{prop:nullSpace}). From \eqref{eq:ARMA_exU}, we see that the inverse relation that appears in \eqref{eq:ARMA_EV_ex} indeed appears as a function affecting the graph eigenmodes. The following proposition concludes this section by stating this finding in a formal way.
%
%
%From~\eqref{eq:ARMA_exU} we can see how the graph filter shapes the frequency spectra of the input signal. The following proposition provide an explicit form of the modal response for the family of \emph{shift-invariant EV-ARMA$_1$ eigenfilters} (SIEVA-GEs).
\begin{myprop}{\textnormal{(Modal Response of SIEVA$_1$)}} An ARMA$_1$ graph filter of the form~\eqref{eq:ARMA_EV_ex} satisfying $(A.1)$ and $(A.2)$ for $K = 1$ has $i$th modal response
\begin{equation}\label{eq:evSIEVA_h}
 h_i = \frac{(\bb_{\U,i}^{\mathcal{A}})^\transp\alphab_{0}}{1-(\bb_{\U,i}^{\mathcal{A}})^\transp\alphab_1}
\end{equation}
where $(\bb_{\Ub,i}^{\A})^\transp$ is the $i$th row of the matrix $\Bs$.
\end{myprop}
\begin{proof}
  The proof follows directly from~\eqref{eq:ARMA_exU}.
\end{proof}
% From~\eqref{eq:evSIEVA_h}, for $\D=\varphi\I$, we can observe the following:
% \begin{remark}\label{rmk:ARMA}
%   The modal response of a SIEVA-GE for the case when the matrix $\D = \varphi\I$ is 
%   \begin{equation}\label{eq:armaEasy}
%     h_i = \frac{\varphi}{1 - (\bb_{\Ub,i}^{\A})^T\alphab}.
%   \end{equation}
% Further, for $\Phib = \psi\bm 1\bm 1^T$ this expression reduces to the classical ARMA$_1$ graph frequency response in~\eqref{eq:ARMAClas}. This result shows once again the generality of the proposed family of filters.
% \end{remark}

\subsection{Filter Design}\label{sub:armaDesg}
%In the following, different designs are presented for EV-ARMA$_1$ graph filters. We first introduce the design for the most general form the node domain. Then we specialize the design for the case of the SIEVA-GEs in the spectral domain to leverage the reduced dimensionality of the resulting optimization problem.
{\textbf{Edge-Variant ARMA$_1$ form.}} Here, we extend the design approach of \cite{isufi2017autoregressivest} and design $\{\Phib_0, \bm{\Phi}_1\}$ by using the Prony's method. For $\tilde\H$ being the desired operator, we can define the fitting error matrix 
\begin{equation}\label{eq:erromtx}
\bm E = \tilde\H - (\bm I - \Phib_1)^{-1}\Phib_0,
\end{equation}
which similar to the classical Prony design presents nonlinearities in the denominator coefficients, i.e., in $\Phib_1$. To tackle these issues, we consider the modified fitting error matrix
\begin{equation}\label{eq:htilde1}
 \bm E' = \tilde\H - \Phib_1\tilde\H -  \Phib_0,
\end{equation}
which is obtained by multiplying both sides of~\eqref{eq:erromtx} by $\bm I - \Phib_1$.

%As minimizing a norm function of~\eqref{eq:erromtx} is hard to handle due to the nonlinearities in the expression, we focus on the minimization of the Frobenius norm of the modified error~\eqref{eq:htilde1} for obtaining the filter coefficients. 
This way, the filter design is transformed in solving the convex optimization problem
\begin{equation}\label{eq:probProny1}
\begin{aligned}
& \underset{\bm\Phi_0,\Phib_1}{\text{minimize}}
& & \Vert\tilde\H - \Phib_1\tilde\H - \Phib_0\Vert\\
& \text{subject to}
& & \Vert\Phib_1\Vert_2 < \delta,~~\delta < 1,\\
& & &{\rm supp}\{\bm \Phi_0\}  = {\rm supp}\{\bm \Phi_1\} = {\rm supp}\{\bm S + \bm I\}.
\end{aligned}
\end{equation}
The objective function in \eqref{eq:probProny1} aims at reducing the modified error $\bm E'$, while the first constraint trades the convergence rate of~\eqref{eq:ARMA_EV} with approximation accuracy.

\textbf{SIEVA$_1$ form.} Following the same idea as in~\eqref{eq:erromtx}-\eqref{eq:probProny1}, the modified fitting error of a SIEVA$_1$ graph filter is
\begin{equation}\label{eq.sieva1err}
e_i' = \tilde h_i - \tilde h_i(\bs)^T\alphab_1 - (\bs)^T\alphab_0,
\end{equation}
with $\tilde h_i$, $(\bs)^\transp\alphab_0$, and $(\bs)^\transp\alphab_1$ denoting respectively the desired modal response and the eigenvalues of $\Phib_0$ and $\Phib_1$ w.r.t. the $i$th mode. In vector form, \eqref{eq.sieva1err} is be written as
\begin{eqnarray}
\bm e' &=& 
\tilde{\bm h}_{\lambda} - \bm \Psi_{\lambda}\bar\alphab\label{eq:linS1},
\end{eqnarray}
with ${\bm e'} = [e_1',\ldots,e_N']^\transp$, $\tilde{\bm h}_{\lambda} = [\tilde{h}_1,\ldots,\tilde{h}_N]^\transp$, $\bm\Psi_{\lambda} = [\Bs,\;\diagg{\tilde{\bm h}_{\lambda}}\Bs]$, and $\bar\alphab = [\alphab_0^\transp\;\alphab_1^\transp]^\transp$. Then, $\{\alphab_0,\alphab_1\}$ can be estimated as the solution of the constrained least squares problem 
\begin{equation}\label{eq:probsieva1}
\begin{aligned}
& \underset{\alphab_0,\alphab_1\in\mathbb{R}^d}{\text{minimize}}
& & \Vert\tilde{\bm h}_{\lambda} - \bm\Psi_{\lambda}{\alphab}\Vert_2^2 \\
& \text{subject to}
& & \Vert \Bs\alphab_1 \Vert_\infty < \delta,~\delta < 1, \alphab = [\alphab_0^\transp\;\alphab_1^\transp]^\transp.
\end{aligned}
\end{equation}
Problem~\eqref{eq:probsieva1} again aims at minimizing the modified fitting error, while tuning the convergence rate through $\delta$.

Differently from the general EV-ARMA$_1$, here the number of unknowns is reduced to $2d$, as now only the vectors $\alphab_0$ and $\alphab_1$ need to be designed. Therefore, due to this low dimensionality, one can also opt for global optimization solvers to find an acceptable local minimum of the true error (i.e., the equivalent of~\eqref{eq:erromtx}).

\begin{remark}
The approximation accuracy of the EV ARMA$_1$ filters can be further improved by following the Shank's method~\cite{shanks1967recursion} used in \cite{isufi2017autoregressive, isufi2017autoregressivest}, or the iterative least-squares approach proposed in~\cite{liu2017filter}. These methods have shown to improve the approximation accuracy of Prony's design by not only taking the modified fitting error into account but also the true one. However, as this idea does not add much to this work, interested readers are redirected to the above references for more details.
\end{remark} 
 
 \begin{table*}
 	\caption{Summary of the different graph filters. $(^*)$ indicates a contribution of this work. Here, ${\rm numIt}$ stands for the maximum number of iterations that the recursion is run.}
	\begin{center}
	\begin{tabular}{|c|c|c|c|c|c|}
		\hline
		\textbf{Filter Type} & \textbf{Expression} & \textbf{Shift-Invariant} & \textbf{Design Strategy} & \textbf{Distributed Costs} & \textbf{Coefficients} \\ [-1em] \\ 
		\hline \\ [-1em]
			Classical FIR~\cite{sandryhaila2013discrete} & 
				$\H_{\text{c}} \triangleq \sum_{k = 0}^K\phi_k\S^k$ & 
				always & 
				LS~\cite{sandryhaila2013discrete}, Chebyshev~\cite{taubin1996optimal,shuman2011distributed} &
				$\mathcal{O}(MK)$ &
				scalars: $\{\phi_k\}$ \\ [-1em] \\
		\hline \\ [-1em]
			NV FIR~\cite{segarra2017optimal} &
				$\H_{\text{nv}} \triangleq \sum_{k = 0}^K{\rm diag}(\bphi_k)\S^k$ &
				not in general &
				LS, convex program~\cite{segarra2017optimal}&
				$\mathcal{O}(MK)$ &
				vectors : $\{\bm{\phi}_k\}$ \\ [-1em]\\ 
		\hline \\ [-1em]
			General EV FIR ($^*$) &
			$\bm{H}_{\rm{ev}} \triangleq \sum\limits_{k=1}^{K}(\bm\Phi_{k}\ldots\bm\Phi_{1})$ &
			not in general &
			iterative design [Sec.~\ref{sec.FIR-EV_design}] &
			$\mathcal{O}(MK)$ &
			matrices : $\{\bm\Phi_k\}$ \\
		\hline \\ [-1em]
			SIEV FIR ($^*$) &
			%$\Hb_{\rm siev} = \Ub\bigg[ \sum\limits_{k=1}^{K}\prod_{j=1}^{k}\diagg{\Bs\alphab_j} + \diagg{\Bs\alphab_{0}} \bigg]\Ub^{-1}$ &
			~\eqref{eq:evFreq}&
			always &
			iterative design [Sec.~\ref{sec.FIR-EV_design}] &
			$\mathcal{O}(MK)$ &
			vectors : $\{\alphab_k\}$ \\ [-1em] \\
		\hline	\\ [-1em]				
			CEV FIR ($^*$) &
			$\H_{\rm{cev}} \triangleq \sum\limits_{k=1}^{K}\bm \Phi_k  \bm S^{k-1}$ &
			not in general &
			LS [Sec.~\ref{subsec:CEV_desi}] &
			$\mathcal{O}(MK)$ &
			matrices :$\{\bm\Phi_k\}$ \\ [-1em] \\
		\hline \\ [-1em]
			SICEV FIR ($^*$) &
			~\eqref{eq:evD} &
			always &
			LS [Sec.~\ref{subsec:CEV_desi}] &
			$\mathcal{O}(MK)$ &
			vectors :$\{\alphab_k\}$ \\ [-1em] \\
		\hline \\ [-1em]
			Classical ARMA$_1$~\cite{isufi2017autoregressive} &
			$\H_{{\rm arma}_{1}} \triangleq \varphi(\I - \psi\S)^{-1}$ &
			always &
			closed-form, iterative design~\cite{isufi2017autoregressive} &
			$\mathcal{O}({\rm numIt}\cdot M)$ &
			scalars : $\{\varphi,\psi\}$\\ [-1em] \\
		\hline \\ [-1em]
			EV ARMA$_1$ ($^*$) &
			$\H_{{\rm evarma}_{1}} \triangleq (\I - \Phib_1)^{-1}\Phib_0$ &
			not in general &
			two-step design [Sec.~\ref{sub:armaDesg}] &
			$\mathcal{O}({\rm numIt}\cdot M)$ &
			matrices : $\{\Phib_0,\Phib_1\}$ \\ [-1em] \\
		\hline \\ [-1em]
			SIEVA$_1$ ($^*$) &
			%$\H_{{\rm sievarma}_{1}} \triangleq (\I - \Phib_1)^{-1}\Phib_0$ &
			~\eqref{eq:ARMA_exU}&
			always &
			two-step design [Sec.~\ref{sub:armaDesg}] &
			$\mathcal{O}({\rm numIt}\cdot M)$ &
			vectors : $\{\alphab_0,\alphab_1\}$ \\ [-1em] \\
		\hline
	\end{tabular}
	\end{center}
	\label{tab.filts}
\end{table*}

\section{Numerical Results}
\label{sec:num}
 
We now present a set of numerical examples to corroborate the applicability of the proposed filters for several distributed tasks. For convenience, Table~\ref{tab.filts} presents a summary of the different graph filters mentioned in this work along with their specifications. In our simulations\footnote{The code to reproduce the figures in this paper can be found at https://gitlab.com/fruzti/graphFilterAdvances}, we made use of the GSP toolbox~\cite{perraudin2014gspbox}.

%The generalization of graph filters through edge-variant graph filters and its benefits are illustrated with numerical experiments on synthetic as well as real datasets for distinct applications within graph signal processing and distributed array signal processing. 
  %-+-+-+-+-+-+-+-+-+-+-+-+-+-+-+-+-+-+-+-+-+-+-+-+-+-+-+-+-+-+-+-+-+-+-+-+-+-+-
 %  \subsection*{Graph Filter Approximation}
	% \begin{figure}[t]
	% 	  \centering
	% 	  \includegraphics[width=0.5\textwidth]{figures/compGraphFilterApproximationGraph}
	% 	  \caption{Community graph with $N = 256$ nodes used for examples of approximation of graph filters.}
	%   \label{fig:comGraph}
	% \end{figure}
%-+-+-+-+-+-+-+-+-+-+-+-+-+-+-+-+-+-+-+-+-+-+-+-+-+-+-+-+-+-+-+-+-+-+-+-+-+-+-
\subsection{Graph Filter Approximation}
We here test the proposed FIR graph filters in approximating a user-provided frequency response. We consider a random community graph of $N=256$ nodes and shift operator $\S = \L$. The frequency responses of interest are two commonly used responses in the GSP community, i.e.,
\begin{itemize}
	\item [$(i)$] the exponential kernel
	$$\tilde{h}(\lambda) \coloneqq e^{-\gamma(\lambda - \mu)^2},$$
	with $\gamma$ and $\mu$ being the spectrum decaying factor and the central parameter respectively;
	\item [$(ii)$] the ideal low-pass filter
	$$
	\tilde{h}(\lambda) = \begin{cases}
	1 & 0 \leq \lambda \leq \lambda_{\text{c}} \\
	0 & \text{otherwise},
	\end{cases}
	$$ 
	with $\lambda_{\text{c}}$ being the cut-off frequency.
\end{itemize}
The approximation accuracy of the different filters is evaluated in terms of the normalized squared error $\text{NSE} = \Vert\tilde{\bm{H}} - \bm H_{\rm{fit}}\Vert_F^2/\Vert\tilde{\bm{H}}\Vert_F^2$. $\bm H_{\rm{fit}}$ stands for the filter matrix of the fitted filters.
\begin{figure*}[t]
		  \centering
		  \begin{subfigure}[b]{0.33\linewidth}
		  	  \psfrag{NMSE}{\small NSE}
		  	  \psfrag{Filter Order [k]}[tc][cc]{\small Filter Order $[K]$}
			  \psfrag{CEV}{\fontsize{8}{8}\selectfont{CEV}}
			  \psfrag{Classical FIRaaa}{\fontsize{8}{8}\selectfont{Classical FIR}}
			  \psfrag{NVa}{\fontsize{8}{8}\selectfont{NV}}
			  \psfrag{SIEVaa}{\fontsize{8}{8}\selectfont{SIEV}}
			  %\psfrag{CEV-Filter}{\tiny CEV FIR}
			  %\psfrag{FIR}{\tiny Classical FIR}
			  %\psfrag{NV-Filter (k) Exponential Kernel}{\tiny NV FIR}
		      %\psfrag{SIEV-Filter (k) Exponential Kernel}{\tiny SIEV FIR}
			  \includegraphics[width=1.1\textwidth]{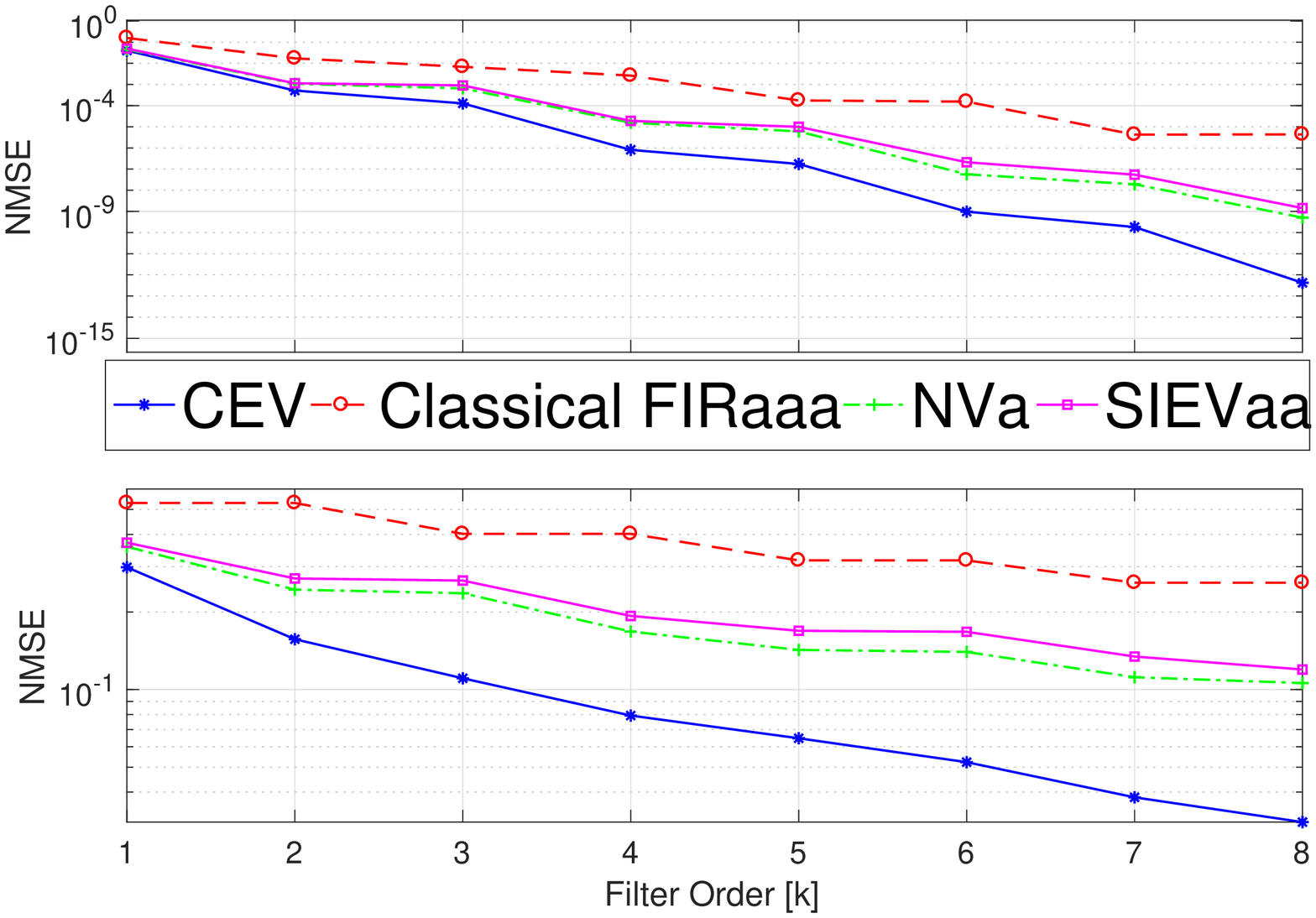}
			  \caption{}
			  \label{}
		  \end{subfigure}%
		  \begin{subfigure}[b]{0.33\linewidth}
		  	  \psfrag{h(l)}[tc][tc]{\tiny $h(\lambda)$}	
		  	  \psfrag{Eigenvalues}[tc][cc]{\small Eigenvalues $\lambda$}
			  \includegraphics[width=\textwidth]{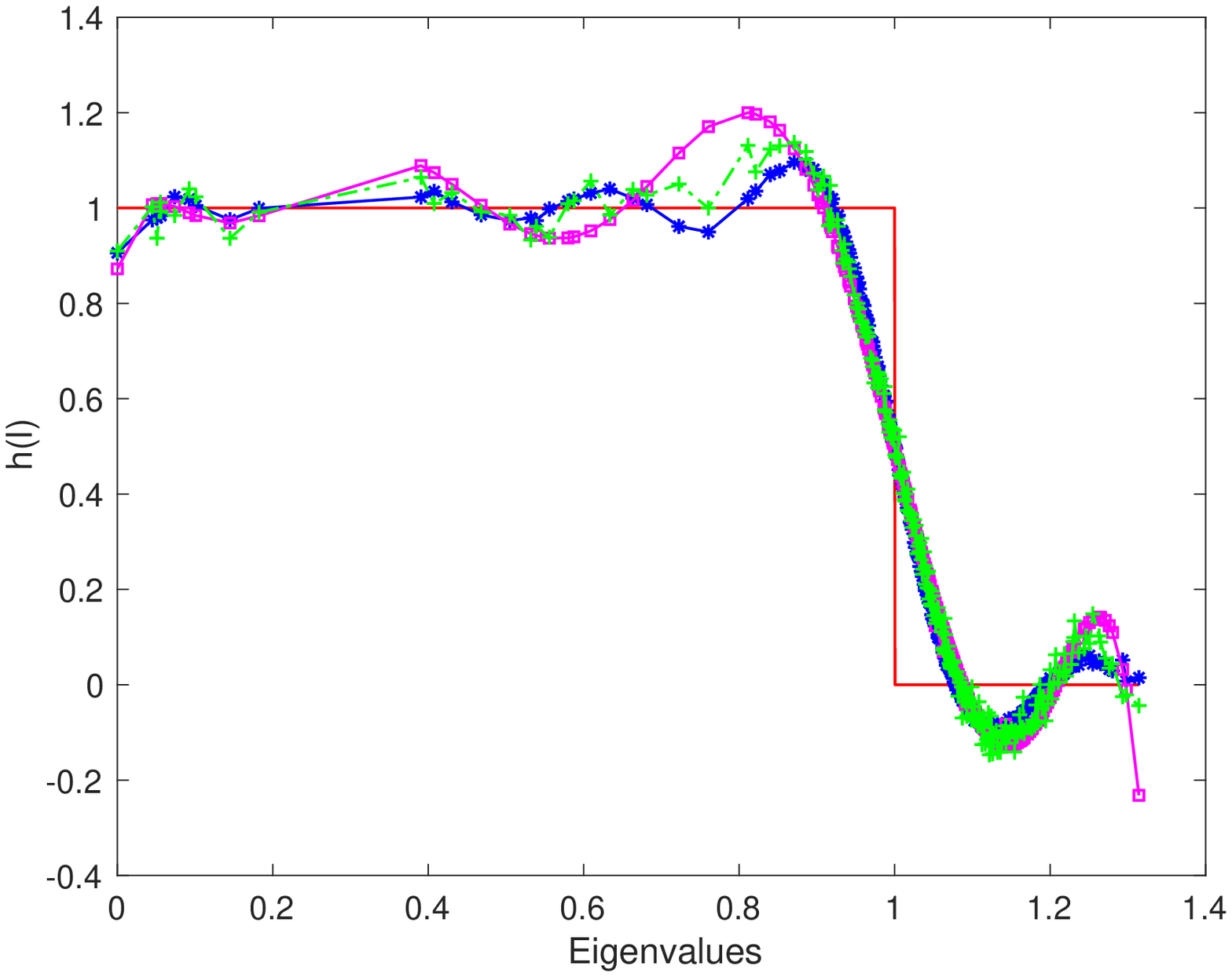}
			  \caption{}
			  \label{}
		  \end{subfigure}%
		  \begin{subfigure}[b]{0.33\linewidth}
		  	  \psfrag{h(l)}[tc][tc]{\tiny $h(\lambda)$}	
		  	  \psfrag{Eigenvalues}[tc][cc]{\small Eigenvalues $[\lambda]$}
			  \includegraphics[width=\textwidth]{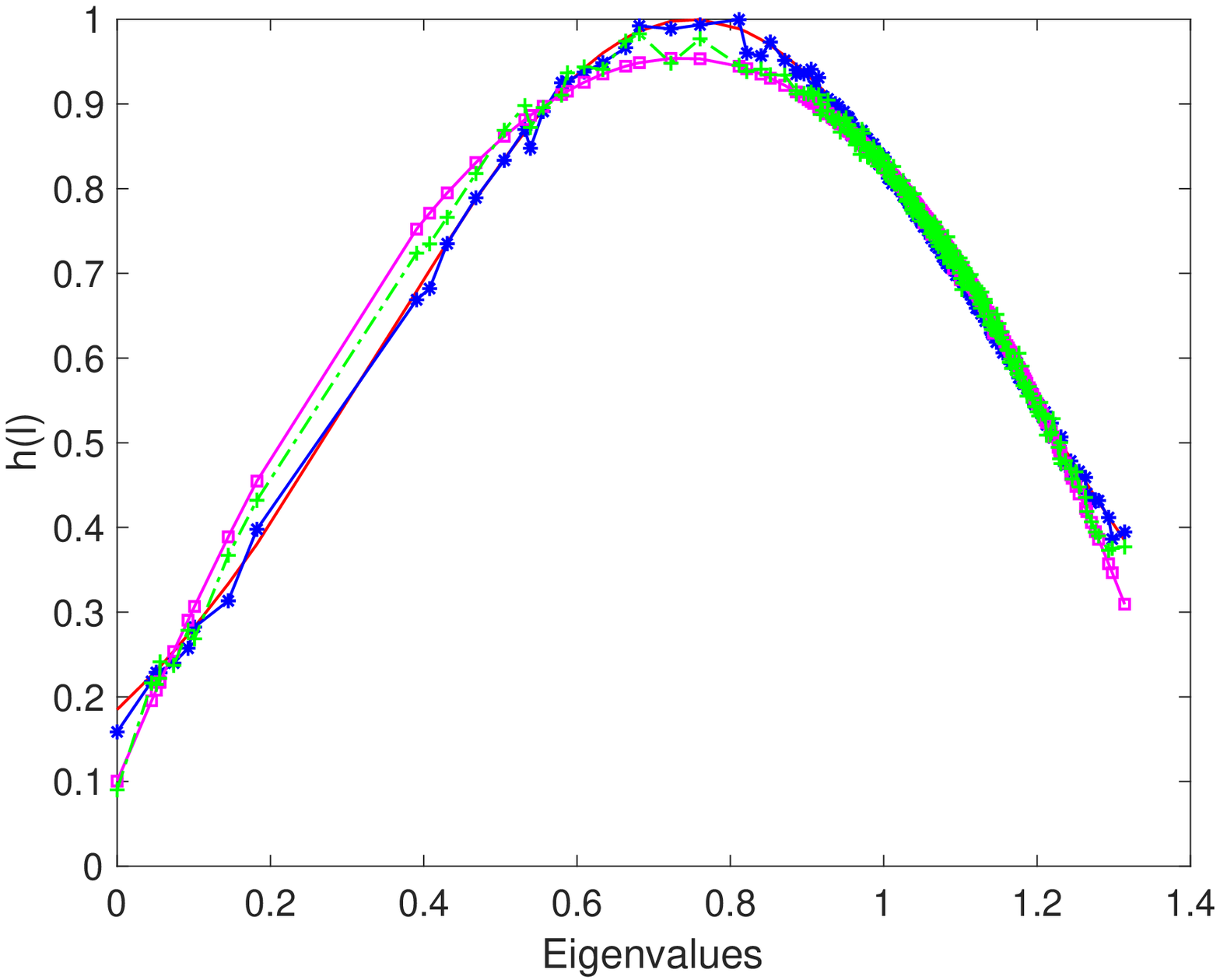}
			  \caption{}
			  \label{}
		  \end{subfigure}	  		  
%		\caption{(a) Error comparison between the proposed EV graph filters, the NV graph filter and the classical FIR for different orders. (Top) Performance for the exponential kernel. (Bottom) Performance in the ideal low pass scenario.}
		\caption{(a) NSE vs. filter order for different FIR graph filters. (Top) Results in approximating a low-pass response. (Bottom) Results in approximating the exponential kernel response. (b) Frequency response of the graph filters when approximating a perfect low pass filter. (c) Frequency response of the graph filters when approximating an exponential kernel with parameters $\mu = 0.75$ and $\gamma = 3$.}
	  \label{fig:compAppxFilt}
	\end{figure*}
	
Fig.~\ref{fig:compAppxFilt} illustrates the performances of the different filters. In the exponential kernel scenario, we observe that the CEV FIR filter outperforms the other alternatives by showing a performance improvement of up to two orders of magnitude. A similar result is also seen in the low-pass example, where the CEV FIR filter achieves the error floor for $K=8$, while the NV graph filter for $K = 13$ and the \emph{classical} FIR filter for $K = 17$. Additionally, we observe that the SIEV FIR filter achieves the same performance as the NV FIR filter. This result suggests that despite the additional DoF of the SIEV FIR filter, the nonconvex design strategy \eqref{eq:cstMinSIEV} yields in a local minimum that does not exploit the filter full capabilities.

%In the low pass scenario, we remark that the CEV graph filter achieves the error floor for $K=8$, while the NV graph filter for $K = 13$ and the \emph{classical} FIR for $K = 17$. In the exponential kernel scenario, we observe again that the constraint EV outperforms the classical and NV graph filters, showing a performance gap of two orders of magnitude. 

%As seen from this results, it is clear that even though the SIEV graph filter has many more degrees of freedom, the tendency of the optimization procedures to only obtain a local optima leads to performance comparable to the one of the NV graph filter. In these examples, the performance of the SIEV and NV graph filters are close to each other. 

The above observations further motivate the use of the CEV FIR filter, which trades off better the simplicity of the design and the available DoF. In fact, even though the CEV FIR filter has less DoF than the SIEV graph filter, it performs better than the latter. Similarly, the larger DoF of the CEV FIR filter compared to the NV FIR filter (i.e., $\rm{nnz}(\textbf{S})\cdot K + N$ vs $N\cdot (K+1)$) allow the CEV FIR filter to better approximate the desired response. In a distributed setting, these benefits translate into communication and computational savings.

%In turn, these characteristics are translated in savings in communication and computational complexity compared to a distributed implementation of classical FIR graph filters. %However, due to this design freedom, we would like to remark that the filter design for the C-EV graph filters might suffer from numerical issues for large filter orders due to the conditioning of the matrix $\bm \Psi$, i.e., the number of parameters to be estimated $(\rm{nnz}(\S)\cdot K + N)$ becomes larger than the effective rank of $\bm\Psi$. Similar problems might arise for NV filters of high order, although the number of parameters to estimate are smaller compared to the C-EV graph filter.
  %-+-+-+-+-+-+-+-+-+-+-+-+-+-+-+-+-+-+-+-+-+-+-+-+-+-+-+-+-+-+-+-+-+-+-+-+-+-+-
  \subsection{Distributed Linear Operator Approximation}

Several distributed tasks of interest consist of performing a linear operation $\bm{A}\in\mathbb{R}^{N\times N}$ over a network. This can be for instance a beamforming matrix over a distributed array or consensus matrix. In most of these cases, such linear operators cannot be straightforwardly distributed. In this section, we illustrate the capabilities of the developed graph filters in addressing this task.

Given a desired linear operator $\bm{A}$, we aim at implementing this linear operator distributively through the solution of the optimization problem
\begin{equation}\label{eq:probLinop}
\begin{aligned}
& \underset{\bm \theta}{\text{minimize}}
& & \Vert \bm{A} - \bm H(\bm S, \bm \theta) \Vert \\
& \text{subject to}
& & \bm \theta \in \Theta,\\
\end{aligned}
\end{equation}
where $\bm H(\bm S, \bm \theta)$ stands for the used graph filter parametrized by the shift $\bm S$ and a set of parameters $\bm \theta$ living in the domain $\Theta$.
  
\textbf{Distributed Consensus.} For distributed consensus, the operator ${\bm{A}}$ has the form ${\bm{A}} = \frac{1}{N}\bm 1\bm1^T$, which for $\bm S = \bm L$ translates into a low-pass graph filter passing only the DC signal component.
	
Fig.~\ref{fig:consensusFig} compares the fitting $\text{NSE} = \Vert\bm{A} - \bm H_{\rm{fit}}\Vert_F^2/\Vert\tilde{\H}\Vert_F^2$ for the different FIR graph filters. We note once again that the CEV implementation offers the best approximation accuracy among the contenders achieving an NSE of order $10^{-4}$ in only $10$ exchanges. These results yield also different insights about the SIEV and SICEV graph filters.

First, both the SIEV and the SICEV implementations fail to compare well with the CEV, though the linear operator $\bm{A}$ is shift invariant. We attribute this degradation in performance to assumption (A.1) necessary for these filters to have a modal response interpretation. In fact, forcing each filter coefficient matrix to be shift invariant seems limiting the filter ability to match well the consensus operator.
	
Second, the different design strategies used in SIEV and SICEV further discriminate the two filters. We can see that the least squares design of the SICEV implementation is more beneficial, though the SIEV filter has more DoF. Unfortunately, this is the main drawback of the latter graph filter, which due to the nonconvexity of the design problem leads to suboptimal solutions. However, we remark that both these filters outperform (or compare equally with) the classical FIR filter. Further investigation in this direction is needed to understand if the SIEV, or SICEV structures can be used to achieve finite-time consensus as carried out in \cite{wang2010finite, sandryhaila2014finite}.

%In Fig.~\ref{fig:consensusFig}  the normalized approximation error, $e_n=\Vert\tilde{\bm{H}} - \bm H_{\rm{fit}}\Vert_F^2/\Vert\tilde{\bm{H}}\Vert_F^2$, is shown for the four shift invariant FIR filters. From this plot, it can be noticed that the CEV graph filter obtains the best the approximation accuracy among all the FIR filters. Furthermore, similar to the previous examples, even though the SIEV graph filter has many more degrees of freedom and a more general structure than the CEV graph filter, its achieved performance is similar to the classical FIR filter. Unfortunately, this is the main drawback of the more general (and complex) structure of the SIEV graph filter, which due to the lack of global optimality in the solution of the filter coefficients in the design problem, only local optimal solutions (suboptimal) can be obtained.

%	\begin{figure}[t]
%		  \centering
%		  \includegraphics[width=0.5\textwidth]{figures/testConsensusEV}
%		  \caption{NMSE versus filter order for different distributed FIR implementations when approximating the consensus operator $\bm H = 1/N\bm 1\bm1^T$.}
%	  \label{fig:consensusFig}
%	\end{figure}
	\begin{figure}[t]
		  \centering
%		  \psfrag{ax}[bc][tc]{$\Vert \tilde{\H} - \H_{\rm fit}\Vert_{\rm F}/\Vert \tilde{\H} \Vert_{\rm F}$}
		  \psfrag{ax}[bc][tc]{NSE}
		  \psfrag{bx}[tc][bc]{Filter Order $[K]$}
		  \psfrag{Classical FIR}{\small Classical FIR}
		  \psfrag{CEV}{\small CEV FIR}
		  \psfrag{NV}{\small NV FIR}
		  \psfrag{SICEV}{\small SICEV FIR}
		  \psfrag{SIEV}{\small SIEV FIR}
		  \includegraphics[width=0.5\textwidth]{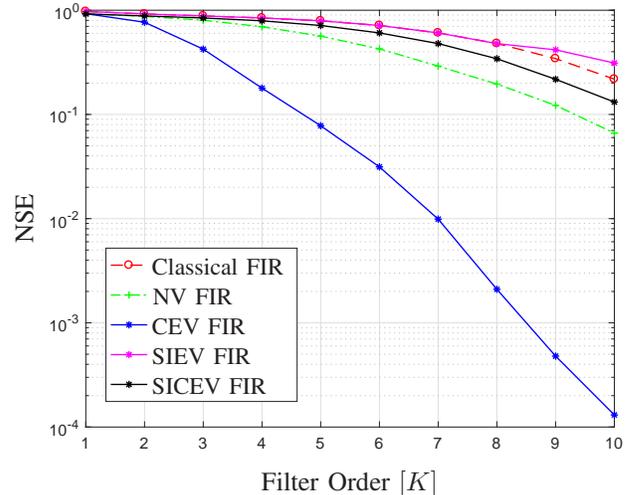}
		  %\caption{Error comparison between the proposed EV graph filters, the NV graph filter and the classical FIR for different orders when approximating a consensus operation, i.e., $\bm H = 1/N\bm 1\bm1^T$.}
		  \caption{NSE versus filter order for different distributed FIR filter implementations when approximating the consensus operator $\bm H = 1/N\bm 1\bm1^T$.}
	  \label{fig:consensusFig}
	\end{figure}
{
\textbf{Wiener-based denoising.} 
For instances when the statistics of the graph signal and noise signal are available, a typical approach for performing denoising is the Wiener filter. This filter is obtained by minimizing the mean-squared error, i.e.,
	\begin{equation}\label{eq.wienerProb}
		\tilde{\H} = \underset{\H\in\mathbb{R}^{N\times N}}{\rm argmin} \mathbb{E}\big[\Vert \H\bm z - \bm x\Vert_{2}^{2}\big],
	\end{equation}
where $\bm z = \bm x + \bm n$ is the graph signal corrupted with additive noise. For the case of zero-mean signals $\bm x$ and $\bm n$ with covariance matrices $\bm\Sigma_{\bm x}$ and $\bm\Sigma_{\bm n}$, respectively, the solution for~\eqref{eq.wienerProb} is 
	\begin{equation}\label{eq.wienerH}
		\tilde{\H} = \bm\Sigma_{\bm x}(\bm\Sigma_{\bm x}+\bm\Sigma_{\bm n})^{-1},
	\end{equation}
given $\bm\Sigma_{\bm x} + \bm\Sigma_{\bm n}$ is not singular. When the covariance matrices $\bm\Sigma_{\bm x}$ and $\bm\Sigma_{\bm n}$ share the eigenvectors with the graph shift operator, the optimal filter $\tilde\H$ can be applied through classical graph filters. However, in many instances, the signal covariance matrix $\bm\Sigma_{\bm x}$ is not diagonalizable by the eigenvectors of $\bm S$~\cite{marques2017stationary}. When $\tilde{\H}$ is not jointly diagonalizable, a typical approach is to consider only the diagonal entries of its projection onto the eigenvectors of the shift operator, $\bm{D} = {\bm U^{-1}\tilde{\H}\bm U}$. Then, a filter $\bar{\H} = \bm U\diagg{[\bm D]_{11},\ldots,[\bm D]_{NN}}\bm U^{-1}$ is used instead of $\tilde{\H}$ as an approximation. For cases where $\bm D$ is approximately diagonal this is a good way to approximate the Wiener filter in a distributable manner. However, for general matrices $\tilde{\H}$ this is not a necessary good approach. 

We illustrate an example where instead of approximating the Wiener filter through a classical FIR graph filter, we employ a CEV FIR filter. For this example we consider
the Molene dataset\footnote{Access to the raw data through the link {donneespubliques.meteofrance.fr/donnees\_libres/Hackathon/RADOMEH.tar.gz}}, where the temperature data of several cities in France has been recorded. The graph employed is taken from~\cite{perraudin2017stationary} and the graph signal has been corrupted with white Guassian noise. The results in terms of NSE for the different fitted graph filters are shown in Fig.~\ref{fig.wienerFilt}. From this plot we observe that the CEV FIR filter outperforms all the other alternatives. This is due to the fact that the optimal Wiener filter is not jointly diagonalizable with the eigenbasis of the shift operator, i.e., covariance matrix of data is not shift invariant, hence classical graph filters are not appropriate to approximate the filter.

\begin{figure}[t]
		  \centering
%		  \psfrag{error}[bc][tc]{$\Vert \tilde{\H} - \H_{\rm fit}\Vert_{\rm F}/\Vert \tilde{\H} \Vert_{\rm F}$}
		  \psfrag{error}[bc][c]{NSE}		  
		  \psfrag{Filter Order [K]}[cc][bc]{Filter Order $[K]$}
		  \psfrag{Classical FIR}{\small Classical FIR}
		  \psfrag{CEV FIR}{\small CEV FIR}
		  \psfrag{NV FIR}{\small NV FIR}
		  \includegraphics[width=0.5\textwidth]{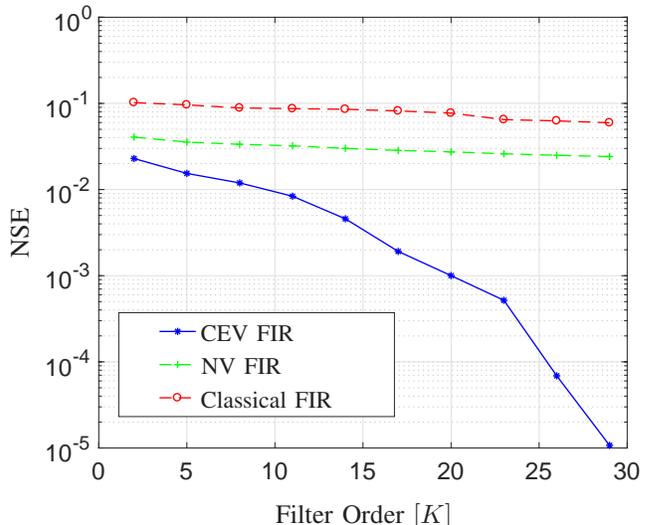}
		  %\caption{Error comparison between the proposed EV graph filters, the NV graph filter and the classical FIR for different orders when approximating a consensus operation, i.e., $\bm H = 1/N\bm 1\bm1^T$.}
		  \caption{NSE versus filter order for different distributed FIR filter implementations when approximating the Wiener Filter for the Molene temperature dataset.}
	  \label{fig.wienerFilt}
	\end{figure}
}	
\textbf{Distributed Beamforming.}
We here consider the task of applying a beamforming matrix $\bm W^\herm$ to signals acquired on a distributed array. More specifically, we aim at obtaining the output
		\begin{equation} \label{eq:beamForming}
			\bm y = \bm W^{\herm}\bm x,
		\end{equation}
where $\bm x$ is the data acquired in a distributed way. Since $\bm W^\herm$ might often be a dense matrix, e.g., in zero-forcing beamforming, operation \eqref{eq:beamForming} cannot be readily distributed. To obtain the output at each node, we approximate the beamforming matrix with different graph filters. 

We quantify this scenario in a distributed $2$D sensor array. The network is generated using $N=40$ random locations on a $2$D plane where the communication network is an $8$-nearest neighbors graph. The beamforming matrix is the matched filter~\cite{harry2002optimum} matrix for a uniform angular grid of $N=40$ points in the range $(-180,180]$. In other words, every node will see the information from a small sector of approximately nine degrees. Since in general $\bm W^{\herm}$ does not share the eigenbasis with $\bm S$, classical graph filters fail to address this task. Therefore, here we compare only the CEV FIR filter and the NV FIR filter. Fig.~\ref{fig:beamPattern} shows two output beampatterns obtained by solving~\eqref{eq:probLinop} with $\bm A = \bm W^\herm$ for the two considered filters with order $K = 5$. We notice that the CEV outperforms the NV FIR filter as it follows more closely the desired beampattern. 

Note that the above framework treats the distributed beamforming differently from approaches based on distributed optimization tools~\cite{sherson2016distributed}. The latter methods usually aim at computing the beamforming matrix (i.e., the weighting matrix is data dependent) and then perform consensus. On the other hand we assume that $\bm W^\herm$ is fixed and that it must be applied to the array data. However, this problem can also be solved through distributed convex optimization tools by solving the least squares problem
		\begin{equation}\label{eq:cvxDist}
			\begin{array}{ll}
				\underset{\bm y}{\text{minimize}} & \Vert \bm x - (\bm W^{\herm})^{\dagger}\bm y \Vert_2^2.
			\end{array}
		\end{equation}
		Differently from~\eqref{eq:probLinop}, formulation avoids the computation of the pseudo-inverse and the graph-filtering based approach requires only five iterations to compute the final beampattern.
		
		In the upcoming section, we compare the CEV and the NV graph filters with distributed optimization tools in solving a general inverse problem.

		% \begin{figure}[t]
		%   \centering
		%   \includegraphics[width=0.5\textwidth]{figures/numSection/beamFormingArray}
		%   \caption{Distributed network example with $N=40$. Sensors are located at random locations and the communication graph is generated as with its kNN graph, with $k=8$.}
	 %  	  \label{fig:networkBeamForming}
		% \end{figure}

%	\begin{figure}[t]
%		  \centering
%		  \begin{subfigure}[b]{\linewidth}
%		  	\psfrag{db}[bc][tc]{\small dB$_{10}$}
%		  	\includegraphics[width=\textwidth]{figures/numSection/beamForming0_legends}
%		  	\caption{}
%		  \end{subfigure}
%		  \begin{subfigure}[b]{\linewidth}
%		    \psfrag{db}[bc][tc]{\small dB$_{10}$}
%		  	\includegraphics[width=\textwidth]{figures/numSection/beamForming90_legends}
%		  	\caption{}
%		  \end{subfigure}
%		  \caption{Beampatterns comparision for different node outputs and the desired steering pattern, $\theta_0$. Angle of interest (a) $\theta_0 = 0^o$, (b) $\theta_0 = 90^o$. }
%	  \label{fig:beamPattern}
%	\end{figure}	
%	    %-+-+-+-+-+-+-+-+-+-+-+-+-+-+-+-+-+-+-+-+-+-+-+-+-+-+-+-+-+-+-+-+-+-+-+-+-+-+-
\begin{figure}[t]
		  \centering
		  \begin{subfigure}[b]{\linewidth}
		  	\psfrag{db}[bc][tc]{\small dB$_{10}$}
		  	\psfrag{Desired Beampattern}{\small Target Beampattern}
		  	\psfrag{Beampattern Hev}{\small CEV FIR Output}
		  	\psfrag{Beampattern Hnv}{\small NV FIR Output}
		  	\psfrag{[o]}[cc][bc]{Angle $[^o]$}
		  	\includegraphics[width=\textwidth]{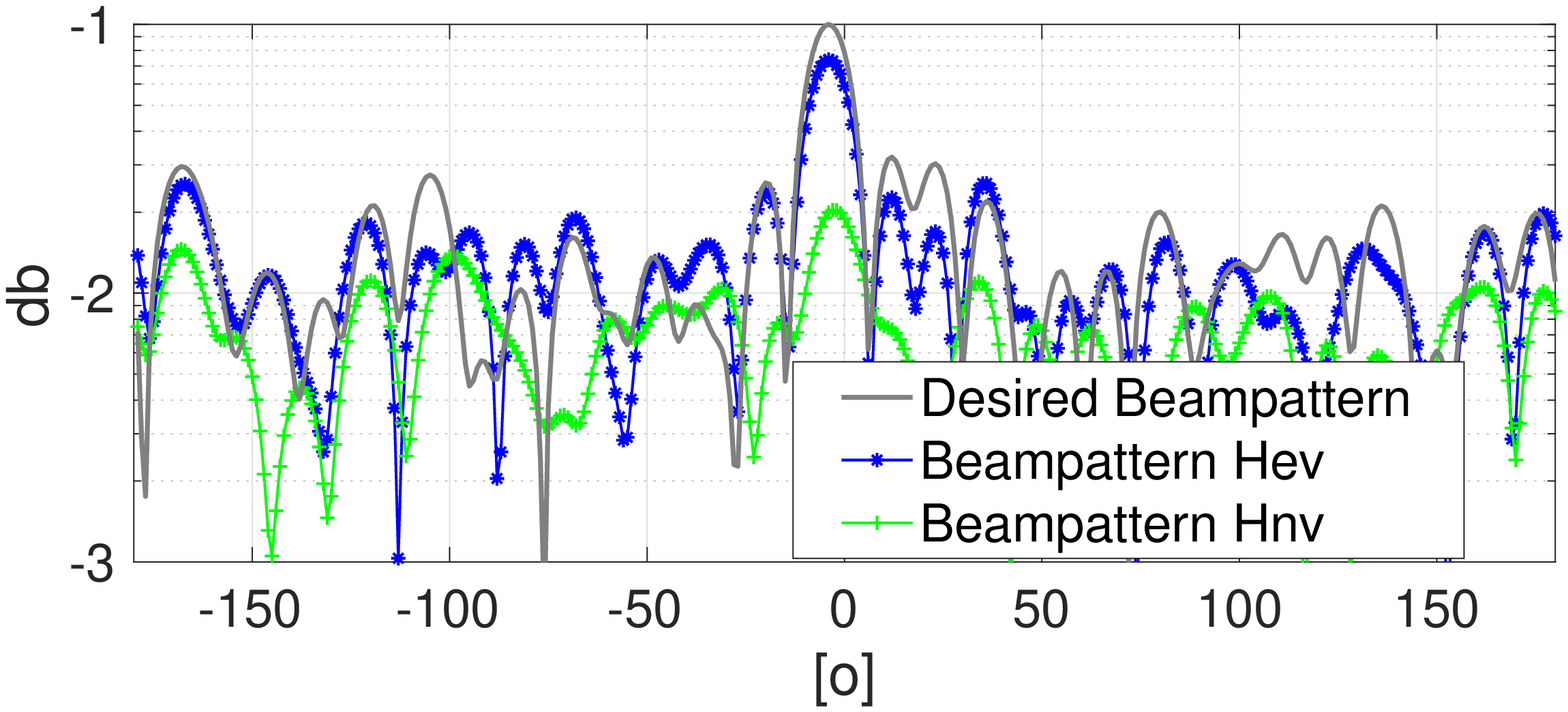}
		  	\caption{}
		  \end{subfigure}
		  \begin{subfigure}[b]{\linewidth}
		  	\psfrag{db}[bc][tc]{\small dB$_{10}$}
		  	\psfrag{[o]}[cc][bc]{Angle $[^o]$}
			\psfrag{Desired Beampattern}{\small Target Beampattern}
		  	\psfrag{Beampattern Hev}{\small CEV FIR Output}
		  	\psfrag{Beampattern Hnv}{\small NV FIR Output}
		  	\includegraphics[width=\textwidth]{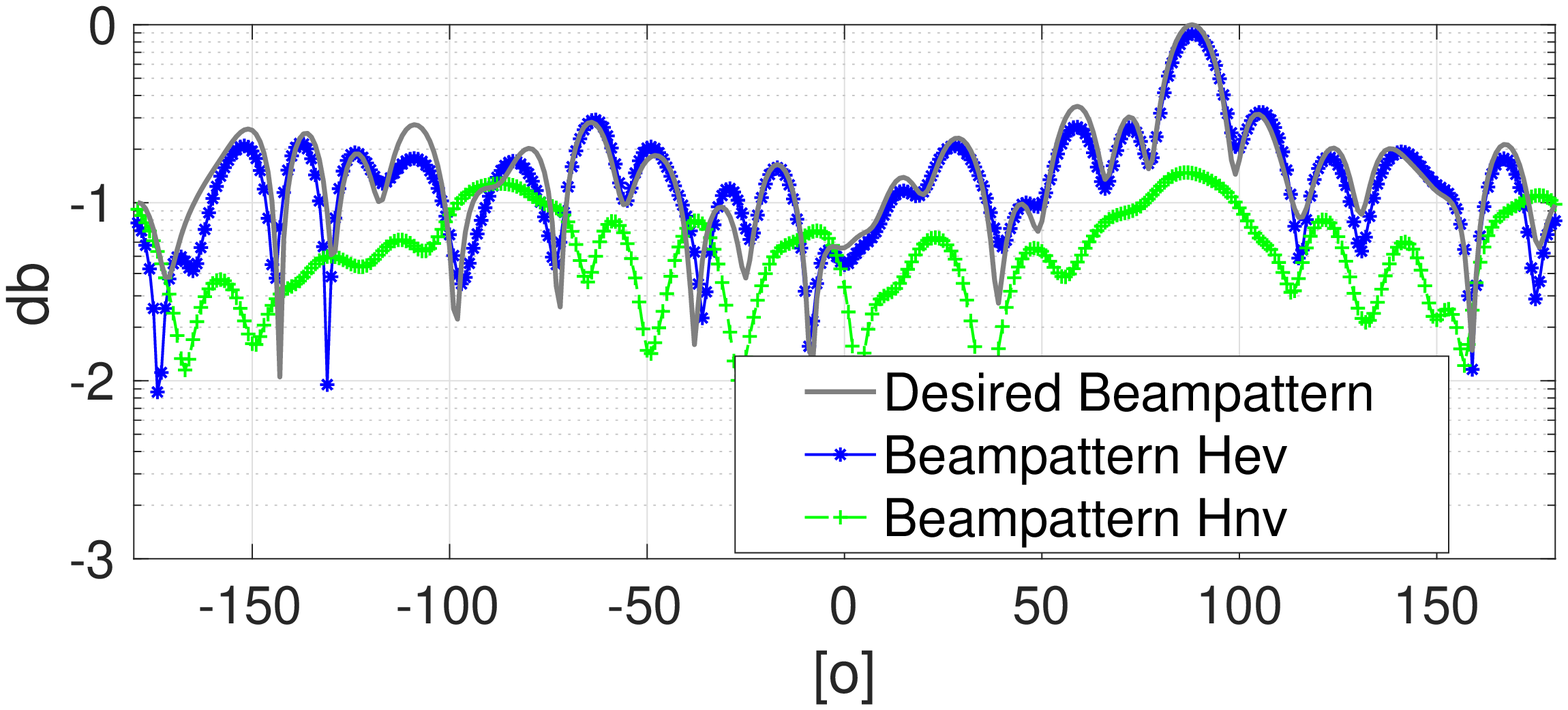}
		  	\caption{}
		  \end{subfigure}
		  \caption{Comparison of beampatterns for different node outputs and desired steering angle, $\theta_0$. Beampattern for beamformer focusing at (a) $\theta_0 = 0^o$ and (b) $\theta_0 = 90^o$. }
	  \label{fig:beamPattern}
	\end{figure}	
  \subsection{Comparison with Distributed Optimization}  
  
  We now compare the proposed graph filters with the primal dual method of multipliers (PDMM)\footnote{PDMM is an alternative distributed optimization tool to the classical alternating direction method of multipliers (ADMM), which is often characterized by a faster converge \cite{zhang2018distributed}.} \cite{zhang2018distributed} solving distributively the least squares problem
  \begin{equation}
  	\begin{array}{ll}
  		\underset{\bm x}{\text{minimize}} & \Vert \bm y - \bm A\bm x \Vert_2^2
  	\end{array}
  \end{equation}	
for a matrix $\bm A$. Without loss of generality we consider $\bm A$ to be an $N \times N$ matrix. The baseline assumption for all distributed optimization methods is that $v_i$ knows its own regressor, i.e., the $i$th row of $\bm A$, $\bm a_i^\transp$. The task is that each node retrieves the full vector ${\bm x}_{\text{ls}} = \bm A^{\dagger}\bm y$ by means of local communications.
 
For the graph filter-based approaches, we approximate $\bm A^{\dagger}$ through a set of rank one matrices $\{\tilde{\bm H}_{i} \triangleq \bm 1\tilde{\bm a}_{i}^{\transp}\}_{i=1}^{N}$ with $\tilde{\bm a_{i}}$ being the $i$th row of $\bm A^{\dagger}$. This means that in contrast to distributed optimization methods, here every node $v_i$ needs to know the full $\bm A$. Each $\tilde{\bm H}_{i}$ is then fitted with the NV and CEV recursions to approximate $\bm x_{\text{ls}}$ as the output after filtering the graph signal $\bm y$. It must be noticed that the number of communications between adjacent nodes does not scale with $N$. In fact, both the NV and the CEV will shift the signal only $K$ times and the nodes can locally apply the respective coefficients to obtain the outputs.

To quantify the performance, we perform $100$ Monte Carlo simulations with a randomly generated system matrix and solution vector. Fig.~\ref{fig:distOptComp} compares the graph filter approaches with the distributed optimization methods in terms of the $\text{NSE} = \Vert \bm x - \hat{\bm x}^{(k)} \Vert_2^{2}/\Vert\bm x\Vert_2^{2}$. The graph filter methods achieve a faster decay compared to the distributed optimization method in the first hundred iterations. However, since perfect approximation of the desired response is not possible both graph filters exhibit an error floor. PDMM, on the other hand, does not run into this issue and guarantees convergence to the true solution. Despite this difference in performance, the graph filter approaches can be employed for cases where the accuracy requirements are not strict, or as \emph{warm starts} for the distributed optimization methods.  The above comparison, besides proposing graph filters as an alternative for solving distributed least squares problems, raises the question on \emph{how graph filters relate to distributed convex optimization}. In fact, further research is needed to relate the design and implementation of distributed EV graph filters with the well-established theory of distributed optimization.

%We believe that this question provides a new venue for deeper research and improvements for the design/implementation of graph filters which directly spans from treating graph filters under the tools and theory of distributed optimization, a direction that until this moment has not been explored.
%
%Notice that even though in this particular example we consider the number of unknowns to be equal to the number of nodes, i.e., $\bm x$ is an $N$-dimensional vector, in general, this does not have to be the case. 

  	\begin{figure}[t]
		  \centering
		  \psfrag{error}[bc][tc]{\small $\Vert\x - \hat{\x}^{(k)}\Vert_2^2/\Vert\bm x\Vert_2^2$}
		  \psfrag{Objective Error Versus Iterations}{}
		  \psfrag{Dist. CVX Opt.}{\fontsize{8}{8}\selectfont{Dist. CVX Opt.}}
		  \psfrag{CEV-GF}{\fontsize{8}{8}\selectfont{CEV FIR}}
		  \psfrag{NV-GF}{\fontsize{8}{8}\selectfont{NV FIR}}
		  \psfrag{Number of Iterations}[tc][cc]{Iteration $[K]$}
		  \includegraphics[width=0.5\textwidth]{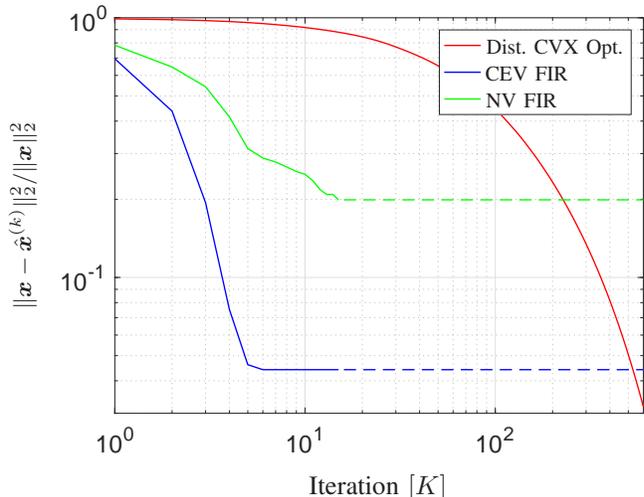}
		  \caption{Convergence error versus the number of iterations for the NV and the CEV graph filters and for the PDMM solver \cite{zhang2018distributed}. Dashed lines indicates the saturation floor of the NV and CEV FIRs.}
	  \label{fig:distOptComp}
	\end{figure}

%   	\paragraph*{Wiener-based denosing} When the statistical properties of both noise and signal on the graph are available, it common to employ a Wiener filter for denoising~\cite{hayes2009statistical,perraudin2017stationary,yaugan2016spectral}. That is, we apply a linear filter such that it minimizes the mean-squared error (MSE), i.e.,
% \begin{equation}
% 	\bm H^{*} = \underset{\bm H\in\mathbb{R}^{N\times N}}{\text{argmin}}\;\;\mathbb{E}\{\Vert\bm H \bm x - \bm z\Vert_{2}^{2}\}
% \end{equation}
% where $\bm z = \bm x + \bm n$ is the graph signal which has been corrupted with additive noise.

    %-+-+-+-+-+-+-+-+-+-+-+-+-+-+-+-+-+-+-+-+-+-+-+-+-+-+-+-+-+-+-+-+-+-+-+-+-+-+-
  \subsection{Tikhonov-based denoising}
  
One of the central problems in GSP is that of recovering an unknown signal $\x$ from a noisy realization $\bm z = \bm x + \bm n$ given that $\bm x$ is smooth w.r.t. the underlying graph \cite{shuman2013emerging}. Differently  known as the Tikhonov denoiser, the estimation of $ \bm x$ can be obtained by solving the regularized least squares problem

%of smoothing a noisy corrupted signal. Differently known as the Tikhonov denoiser \cite{shuman2013emerging}, this problem consists of recovering a signal $\x$ from a noisy realization $\bm z = \bm x + \bm n$ given that $\bm x$ is smooth w.r.t. the underlying graph \cite{shuman2013emerging}. $\bm n$ here denotes the additive noise vector.
  %
  %
  %
%  Consider the problem of recovering an unknown signal $\x$ from a noisy realization $\bm z = \bm t + \bm n$, where $\bm n$ is the additive noise vector. When we consider that $\bm x$ is a graph signal, and that it is smooth w.r.t. a given graph~\cite{dong2016learning} represented by $\bm S$, the denoising can be formulated as the following regularized least squares problem:
%
\begin{equation}\label{eq:tik}
\begin{aligned}
{\bm x}^* = \underset{\bm x\in\mathbb{R}^{N}}{\text{arg min}} & \Vert \bm z - \bm x \Vert_{2}^{2} + \mu\bm x^{T}\bm S\bm x,
\end{aligned}
\end{equation}
for $\bm S = \bm L$ and where $\mu$ trades off the noise removal with the smoothness prior. Problem \eqref{eq:tik} has the well-known solution ${\bm x}^* = (\bm I + \mu\bm S)^{-1}\bm z$, which in terms of the terminology used in Section~\ref{sec:prem} is an ARMA$_1$ graph filter with $\varphi=1$ and $\psi=-\mu$ (see also \cite{isufi2017autoregressive} for further analysis). While recursion \eqref{eq:ARMA} can implement this problem distributively, the convergence of the Neumann series in~\eqref{eq:ARMAss} cannot be controlled as the rate is fixed by $\vert \mu\vert\lambda_{\max}\{\bm S\}$.

Here, we show that through the EV ARMA$_1$~\eqref{eq:ARMA_EV} it is possible to improve the convergence speed of the ARMA$_1$ graph filter by exploiting the additional DoF given by the edge-weighting matrices $\{\Phib_0,\Phib_1\}$. However, since now the design is not exact and involves the modified  error [cf. \eqref{eq:probProny1}], this speed benefit will come at the expense of accuracy.
To illustrate this, we consider an example of problem~\eqref{eq:tik} with $\mu = 0.8$ and $\bm S = \lambda_{\max}^{-1}(\bm L)\bm L$, such that $\bm S$ has unitary spectral norm. Fig.~\ref{fig:armaComp} shows the convergence error of the EV ARMA$_1$ for different values of $\delta$ in \eqref{eq:probProny1} and compares it with the classical ARMA$_1$ and the CEV of order $K = 15$.

We make the following observations. First, low values of $\delta$ are preferred to improve the convergence speed. However, values below $0.7$ should in general be avoided since this restricts too much the feasible set of \eqref{eq:probProny1}, hence leading to a worse approximation error. Second, values of $\delta \approx 0.7$ seem to give the best tradeoff, since the convergence speed is doubled w.r.t the ARMA$_1$ and the approximation error is close to machine precision. Additionally, the fact that the solution $\delta = 0.7$ achieves a better performance than the solution with $\delta = 0.8$ arises from the fact that due to the two-step procedure, the solution obtained by minimizing the modified error might not lead to the best matrix for minimizing the true error during the second step.
%Third, compared to the CEV, the EV ARMA$_1$ is an alternative only when high accuracy (i.e., an error below $10^{-10}$) is required. In fact, even higher order of CEV do not improve much the precision of this task w.r.t. $K = 15$. 
Finally, we did not plot the classical FIR filter for solving this problem, since its performance is identical to the ARMA$_1$ for the same distributed costs \cite{isufi2017autoregressive}.

%From Fig.~\eqref{fig:armaComp} we can observe that for lower values of $\rho$ the decay rate of the error increases. However, due to the restriction in the spectral norm of $\Phib_0$ for some low values of $\rho$ error saturates at a high level. Despite this, we observe that for a $\rho = 0.7$, the EV-ARMA$_1$ achieves a negligible error (close to machine precision) almost $100$ iterations before that the classical ARMA$_1$. As a result, depending on the required accuracy, by appropriate design of the matrices $\{\Phib_0,\Phib_1\}$ graph filters achieving faster steady-state response can be obtained. Improvements on the saturation point of the steady-state response are possible by further refinement during the design stage, e.g., iterative methods or alternative approaches not involving the modified error.

  	\begin{figure}[t]
		  \centering
		  \psfrag{p=0.6}{\small$\delta=0.6$}
		  \psfrag{p=0.7}{\small$\delta=0.7$}
		  \psfrag{p=0.8}{\small$\delta=0.8$}
		  %\psfrag{0}{$0$} 		  \psfrag{200}{$200$}		  \psfrag{400}{$400$}		  
		  %\psfrag{600}{$600$}		  \psfrag{800}{$800$} 	\psfrag{1000}{$1000$}
%		  \psfrag{error}[bc][tc]{\small $\Vert \tilde{\H} - \H_{{\rm fit}_1}^{(k)} \Vert_{\rm F}^2/\Vert \tilde{\H}\Vert_{\rm F}^2$}
		  \psfrag{error}[bc][tc]{NSE}
		  \psfrag{Classical ARMA}{\small{Classical ARMA$_1$}}
		  \psfrag{CEV  K = 15}{\small{CEV FIR, $K=15$}}
		  \psfrag{Number of Iteration [n]}[cc][bc]{Iteration $[K]$}
		  \includegraphics[width=0.5\textwidth]{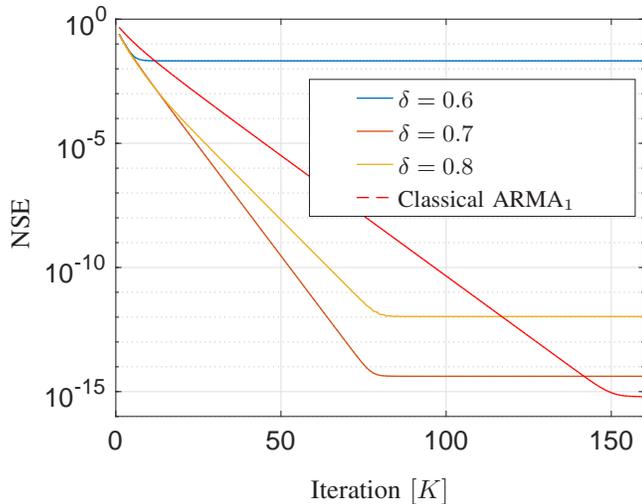}
		  \caption{Convergence error versus the number of iterations for the Tikhonov denoising problem. The EV ARMA$_1$ results are plotted for different values of $\delta$ in \eqref{eq:probProny1} to highlight the tradeoff between convergence speed and approximation accuracy.}
	  \label{fig:armaComp}
	\end{figure}
 %  %-+-+-+-+-+-+-+-+-+-+-+-+-+-+-+-+-+-+-+-+-+-+-+-+-+-+-+-+-+-+-+-+-+-+-+-+-+-+-
 %  \subsection*{Dictionary Learning}
	
	% \begin{figure}[t]
	% 	  \centering
	% 	  \includegraphics[width=0.5\textwidth]{figures/numSection/dictLearningExamp}
	% 	  \caption{Asas}
	%   \label{fig:dictLearning}
	% \end{figure}
 %  %\lipsum[1-12]
 %  %-+-+-+-+-+-+-+-+-+-+-+-+-+-+-+-+-+-+-+-+-+-+-+-+-+-+-+-+-+-+-+-+-+-+-+-+-+-+-
%===============================================================================
\section{Conclusions}
\label{sec:con}
In this work, a generalization of the distributed graph filters was proposed. These filters, that we referred to as edge-variant graph filters, have the ability to assign different weights to the information coming from different neighbors. Through the design of edge-weighting matrices, we have shown that it is possible to weigh, possibly in an asymmetric fashion, the information propagated in the network and improve the performance of state-of-the-art graph filters. 

By introducing the notion of filter modal response, we showed that a subclass of the edge-variant graph filters have a graph Fourier interpretation that illustrates the filter action on the graph modes. Despite that the most general edge-variant graph filter encounters numerical challenges in the design phase, a constrained version of it was introduced to tackle this issue. The so-called constrained edge-variant graph filter still enjoys a similar distributed implementation, generalizes the state-of-the-art approaches, and is characterized by a simple least squares design. For the constrained version, we also showed that there exists a subclass which has a modal response interpretation.

Finally, we extended the edge-variant idea to the family of IIR graph filters, particularly to the ARMA$_1$ graph filter. We showed that by adopting the same local structure a distributed rational filter can be achieved, yet with a much faster convergence speed. Several numerical tests corroborate our findings and show the potential of the proposed filters to improve state-of-the-art techniques.

Future research in this direction should concern the following points: $i)$ improve the design strategy for the more general edge-variant version; $ii)$ improve the saturation accuracy of the proposed methods when dealing with a distributed implementation of linear operators; $iii)$ conciliate the world of GSP with that of distributed optimization and exploit the latter to design distributed graph filters; and $iv)$ extend the edge-variant concept beyond the ARMA$_1$ implementation to the global family of IIR graph filters.
%===============================================================================
% \bibliographystyle{IEEEtran}
% \bibliography{IEEEabrv,bibFile_MACM}
\vspace{-3ex}
\bibliographystyle{IEEEtran}
\bibliography{dissertation}

\end{document}